%% file: main-MOR.tex
\documentclass[moor]{informs1rad}              

\input{tex/macro-MOR}
\begin{document}
	
	
	\RUNAUTHOR{Delong, Farhadi, Niazadeh, Sivan, Udwani}
	
	\RUNTITLE{Online Bipartite Matching with Reusable Resource}
	
	
	\ARTICLEAUTHORS{%
		\AUTHOR{Steven Delong}
		\AFF{Google, \EMAIL{sdelong@google.com}}
		\AUTHOR{Alireza Farhadi}
		\AFF{Carnegie Mellon University, Computer Science, \EMAIL{farhadi94 gmail com}}
				\AUTHOR{Rad Niazadeh}
		\AFF{University of Chicago, Booth School of Business, \EMAIL{rad.niazadeh@chicagobooth.edu}}
				\AUTHOR{Balasubramanian Sivan}
		\AFF{Google Research, \EMAIL{balusivan@google.com }}
		\AUTHOR{Rajan Udwani}
	\AFF{UC Berkeley, IEOR, 
		\EMAIL{rudwani@berkeley.edu}}
	} 
	
	\ABSTRACT{%
	\input{tex/abstract}
	}%
	
	
	\TITLE{Online Bipartite Matching with Reusable Resources}
	
	\maketitle
	
	%
\newif\ifNotes
\Notesfalse

\ifNotes
 \RN{I updated the notes below and left some comments for the next person who takes a pass}
 
 \smallskip
	{\color{red} 
		Notes for the merger
		\begin{itemize}
\item (New?) Title and author affiliations \RN{done}
\item New/modified abstract \RN {done}
\item New Introduction  \RN{done. the new introduction is basically roadmap; we now have two mini introductions at the beginning of sections 3 and 4}
\item Model and notation from OCR paper. \RN{done (still need to add some more details to the model based on Rajan's requests.)}
\item Some requests for notation.
\begin{itemize}
	\item Can we state in the model that ``arrivals" and ``requests" are used interchangeably? \RN{done}
	\item Can we say that $t$/$j$ represent an online vertex as well as its arrival time and simply write $t\in U$ instead of $t\in[1,|U|]$. {\color{blue} done}
	\item (after making the change) I realized that the lower limit $\max\{j-d+1,1\}$ makes some of the equations bulge out in a weird way. Could we write $j-d+1$ for simplicity and explain somewhere that, for example, we can add dummy arrivals before arrival 1? {\color{blue} letting this one go}
\end{itemize}
\item Union of "Related Works" section and merge bib files, do duplicate labels in bib need to be sorted? \RN{done. would be great if someone checks to make sure we have no duplicate entries.} {\color{blue} Rajan: I did another check}
\item Merge primal-dual introduction and modify references in analysis accordingly \RN{done}
			\item Modified conclusion \RN{done}
			\item Merge Acknowledgments {\color{blue} I turned all acknowledgements into a generic ``authors thank". Please feel free to change if desired.}
			\item \RN{The new technical layout merges the two technical sections. Would be great if someone double checks everything. Sections 2,3,4 should now be fully consistent.}
		\end{itemize}
	\hfill\\~\\
	Notation change tally, in case something seems off. \RN{Would be great if someone double checks this! }{\color{blue} Rajan: I did another check}
	\begin{itemize}
		\item $I\to V$
		\item $t\to j$, $T\to U$
		\item $\theta\to\alpha$, $it\to i,t$ or sometimes $i,u$
		\item $\lambda_t\to \beta_j$
		\item $\eta\to \Gamma$
		\item OM to OBM and OMR to OBMRR
		
	\end{itemize}
	}
\fi

\input{tex/intro}

\input{tex/LP}

\input{tex/reranking}

\input{tex/reranking-analysis}

\input{tex/OCR}

\input{tex/primal-dual-OCR}

\input{tex/conclusion}

%
%

%
\input{tex/ack.tex}

\bibliographystyle{informs2014.bst}
\bibliography{refs.bib}

\newpage
\input{tex/appendix.tex}

\end{document}

%% file: tex/macro-MOR.tex
\usepackage[colorlinks=true,breaklinks=true,bookmarks=true,urlcolor=blue,
citecolor=blue,linkcolor=blue,bookmarksopen=false,draft=false]{hyperref}

\def\EMAIL#1{\href{mailto:#1}{#1}}

\usepackage{csquotes}
\usepackage{natbib}
\bibpunct[, ]{(}{)}{,}{a}{}{,}%
\def\bibfont{\small}%
\usepackage{booktabs}

\usepackage[ruled]{algorithm2e} 

\SetAlFnt{\small}
\SetAlCapFnt{\small}
\SetAlCapNameFnt{\small}
\SetAlCapHSkip{0pt}
\IncMargin{-\parindent}

\usepackage{algorithmic}

\newlength\myindent
\usepackage{bbm, color, enumerate, amsbsy, amsfonts, amsopn, amssymb,  amsxtra, bezier, graphicx, latexsym, verbatim,
	pictexwd, supertabular, url, dsfont,leftidx, setspace}
\usepackage{mathtools,bm}
\usepackage{multirow}
\usepackage{footnote}
\makesavenoteenv{tabular}
\makesavenoteenv{table}
\usepackage{cleveref}
\usepackage{enumitem}

\usepackage{tikz}

\DeclareRobustCommand*\cal{\@fontswitch\relax\mathcal}

\usepackage{pgf,pgfplots}
\pgfplotsset{compat=1.15}
\usetikzlibrary{arrows}
\usetikzlibrary{patterns}
\usepgfplotslibrary{fillbetween}
\usetikzlibrary{intersections}
\usetikzlibrary{arrows,decorations.markings}
\usetikzlibrary{positioning}
\tikzset{main node/.style={circle,fill=black!20,draw,minimum size=0.75cm,inner sep=0pt},
            }

\tikzstyle{vecArrow} = [thick, decoration={markings,mark=at position
	1 with {\arrow[semithick]{open triangle 60}}},
	double distance=1.4pt, shorten >= 5.5pt,
	preaction = {decorate},
postaction = {draw,line width=1.4pt, white,shorten >= 4.5pt}]
\tikzstyle{innerWhite} = [semithick, white,line width=1.4pt, shorten >= 4.5pt]

\newcommand{\omrr}{OBMRR}
\newcommand{\ocs}{OCR}
\newcommand{\prev}{\texttt{prev}}

\newcommand{\E}{\mathbb{E}}
\newcommand{\sender}{\texttt{sender}}
\newcommand{\receiver}{\texttt{receiver}}

\DeclareMathOperator{\sel}{selected}
\DeclareMathOperator{\nsel}{not-selected}
\DeclareMathOperator{\un}{unknown}

\tikzset{node/.style={circle,fill=black!100,draw,minimum size=0.75cm,inner sep=0pt, text=white},
            }

\newtheorem{eg}{Example}[section]

\newtheorem{observation}{Observation}

\newcommand{\mb}[1]{\ensuremath{\boldsymbol{#1}}}

\newcommand{\onee}{\mathbbm{1}}

\DeclarePairedDelimiter{\ceil}{\lceil}{\rceil}

\def\alg{\textsc{PR}}

\TheoremsNumberedThrough     

\EquationsNumberedThrough    

\definecolor{cornellred}{rgb}{0.7, 0.11, 0.11}
\definecolor{maroon}{rgb}{0.52, 0, 0}
\definecolor{dgreen}{rgb}{0.0, 0.5, 0.0}
\definecolor{ballblue}{rgb}{0.13, 0.67, 0.8}
\definecolor{royalblue(web)}{rgb}{0.25, 0.41, 0.88}
\definecolor{bleudefrance}{rgb}{0.19, 0.55, 0.91}
\definecolor{royalazure}{rgb}{0.0, 0.22, 0.66}

\newcommand{\RN}[1]{\textcolor{maroon}{ \textbf{/*} \textbf{Rad}: {#1} \textbf{*/}}}

\usepackage{float}

%% file: tex/abstract.tex
We study the classic online bipartite matching problem with a twist: offline vertices, called resources, are \emph{reusable}. In particular, when a resource is matched to an online vertex it is unavailable for a deterministic time duration $d$ after which it becomes available again for a re-match. Thus, a resource can be matched to many different online vertices over a period of time. 
 
While recent work on the problem have resolved the asymptotic case where we have large starting inventory (i.e., many copies) of every resource, we consider the (more general) case of \emph{unit inventory} and give the first algorithms that are provably better than the na\"ive greedy approach which has a competitive ratio of (exactly) 0.5. 

Our first algorithm, which achieves a competitive ratio of $0.589$,  generalizes the classic RANKING algorithm for online bipartite matching of non-reusable resources~\citep{kvv}, by \emph{reranking} resources independently over time. While reranking resources frequently has the same worst case performance as greedy, we show that reranking intermittently on a periodic schedule succeeds in addressing reusability of resources and performs significantly better than greedy in the worst case. Our second algorithm, which achieves a competitive ratio of $0.505$, is a primal-dual randomized algorithm that works by suggesting up to two resources as candidate matches for every online vertex, and then breaking the tie to make the final matching selection in a randomized correlated fashion over time. As a key component of our algorithm, we suitably adapt and extend the powerful technique of online correlated selection~\citep{fahrbach2020edge} to reusable resources, in order to induce negative correlation in our tie breaking step and to beat the competitive ratio of $0.5$. Both of our results also extend to the case where offline vertices have weights.


%

%% file: tex/intro.tex
\section{Introduction}
\label{sec:intro}
 A central task of a two-sided online marketplaces is to match demand to supply. In several applications, the supply is available at the beginning of the decision making horizon --- or equivalently, it is ``offline"--- while the demand arrives sequentially --- or equivalently, it is ``online''. The flagship example for this scenario is in sponsored search, where a market algorithm (for example Google Adwords) decides how to match the arriving search queries to advertisers who are willing to show their ads for that particular arriving query. As an important feature of this example, the advertisers' budgets are considered as \emph{non-reusable} resources: once a query is matched to an advertiser, the consumed budget of that advertiser is gone; in other words, it cannot be reused in the remaining of the decision-making horizon.

In contrast to the above applications, there are prevalent two-sided online marketplaces where the planner aims to match the arriving demand to rental resources. Rental resources are \emph{reusable}, simply because a given unit of such a resource may be re-allocated several times. Examples are virtual machines in cloud computing platforms such as AWS, listings in vacation rental or hospitality service online marketplaces such as Airbnb, and local professional services in online labor platforms such as Thumbtack. Reusable resources are also common in more traditional application domains. One example is in a city hospital where scarce medical resources need to be 
allocated 
for a variety of patient treatments over time
; most of these resources are released to be used again after a patient   leaves. This service operation was particularly important in early days of the COVID-19 pandemic, when hospitals had to carefully manage the allocation of their beds, ventilators, and ICU units to treat a rising influx of COVID patients. 

A common modeling theme in all of the above rental applications is that the planner's algorithm sequentially matches arriving requests to available compatible resources, where a resource is considered available only if it is not under rental at the time of the request. The main focus of this paper is to design and analyze such algorithms with the goal of matching as many demand requests as possible. The online bipartite matching (OBM) problem, originated from the seminal work of Karp, Vazirani and Vazirani~(\citeyear{kvv}), is a prominent model to capture a special case  where the offline resources can be only used once. Motivated by the rental applications, we revisit this classic problem by considering a variation where offline vertices are reusable, that is, once an offline vertex is used in the matching, it will be released to be reused after a certain duration of time. We refer to this problem as the \emph{online bipartite matching with reusable resources (OBMRR)}. This basic model was first introduced in \cite{reuse} -- together with its generalization to online assortment planning of reusable supply. Later it was studied more extensively in \cite{feng,feng3} under deterministic rental durations and in \cite{rba1,full} under stochastic i.i.d. rental durations --- but only in the regime where resources have large capacities (or equivalently fractional allocations are allowed). In \omrr, we study the integral allocation problem where there is exactly one unit of each offline reusable resource available for allocation.

\paragraph{\textbf{Problem Formulation}} Formally, an instance of the \omrr\ problem consists of a bipartite graph $G=(V,U;E)$, where $V$ is the offline side, $U$ is the online side, and $E$ is the set of edges. $V$ represents the set of reusable resources and $U$ represents the set of arriving (demand) requests. Also, an edge $(i,j) \in E$ indicates that resource $i$ can be assigned to request $j$ (also referred to as the arrival $j$). Vertices in $U$ arrive online over discrete time $j=1,2,\ldots,\lvert U\rvert$ in an (oblivious) adversarially chosen order, and reveal their incident edges to vertices in $V$ upon arrival. Once offline vertex $i$ is matched to an arriving online vertex $j$, it remains unavailable for the next $d-1$ periods and will be available for re-allocation upon arrival of online vertex $j+d$, where $d\in\mathbb{N}$ is a known parameter referred to as the rental or usage duration in the paper. 

We consider online algorithms that irrevocably match the arriving online vertex at each time to one of the available offline neighbours (or refuse to match that online vertex). The goal is to maximize the total number of matched online vertices at the end, which we also refer to as the size of the final matching.\footnote{Note that while each online vertex is only matched at most once, each offline vertex can be used multiple times due to reusability of resources.} As is common in the literature on online algorithms,  we compare our algorithm to the optimal offline algorithm that has the complete knowledge of $G$. We say that an algorithm is $\Gamma$-competitive for $\Gamma\in[0,1]$ if for every instance of the problem the expected size of the matching returned by the algorithm is at least $\Gamma$ times the size of the matching produced by the optimal offline solution. We refer to $\Gamma$ as the \emph{competitive ratio}. Important to recall, what makes our problem substantially different from the setting in \cite{kvv} is that each resource $i$ is available at time $j$ if and only if it had not been matched at some time $j'\in[j-d+1,j-1]$. In fact, OBM is a special case of OBMRR when $d=\lvert U \rvert$. 

As mentioned earlier, in this paper we are interested in \emph{integral} online matching algorithms under reusable resources. In contrast to fractional online algorithms that can divide the arriving online vertex and allocate it fractionally to offline vertices, we force the online algorithm -- which can be randomized -- to match the arriving online vertex in full to an offline vertex (or leave it unmatched). We do not assume offline vertices have large budgets; instead, we let every offline vertex to have a unit budget (inventory capacity) as in the classic OBM problem. This is without loss of generality; an instance where resources have non-unit budgets can be transformed into an instance of the unit budget setting (see Proposition 1 in \cite{reuse}). From a practical standpoint, the unit budget setting captures scenarios where the reusable resources have small initial inventories such as in the AirBnB or Thumbtack applications mentioned earlier. This is in contrast to prior work~\citep{feng,feng3,rba1,full} which either consider fractional allocations or the equivalent formulation that asks for integral allocations but with large budgets.

Finally, we highlight that we show our results for the (more general) vertex-weighted setting. In this setting, each offline vertex $i$ is also associated with a non-negative vertex weight $r_i$ and an algorithm for this problem obtains a reward of $r_i$ each time the offline node $i$ gets matched. The goal of both the online algorithm and the optimum offline solution is to maximize the total accumulated weight from matching online nodes. The special case of this model when $d=\lvert U\rvert$ is the non-reusable vertex-weighted online bipartite matching problem introduced and studied in \cite{goel}.

\paragraph{\textbf{Main Contributions}} Any greedy algorithm always matches as many as half of the online vertices matched by the optimum offline solution in the unweighted  setting.\footnote{This folklore result is true for exactly the same reason as a maximal matching in a graph having a cardinality at least half of the size of the maximum matching.} Similarly, the naive greedy algorithm that matches the arriving online vertex to the available offline vertex with maximum weight obtains half of the total weight of the optimum offline solution in the vertex-weighted setting. No deterministic integral algorithm can beat this $0.5$ competitive ratio, which is even true in the special case of unweighted non-reusable resources. In this special case, the elegant RANKING algorithm in \cite{kvv} uses a randomized (negatively) correlated tie breaking rule for greedy matching through one uniform random permutation over offline vertices; the resulting structured correlation improves the competitive ratio from $0.5$ to the optimal  $\left(1-1/e\right)\approx 0.63$. For the vertex-weighted setting with non-reusable resources, \cite{goel} extend RANKING by proposing the Perturbed Greedy (PG) algorithm and obtain the optimal $\left(1-1/e\right)$ competitive ratio. At the same time, \cite{feng,feng3,rba1,full} consider the fractional  \omrr\ problem (with vertex weights) and show adaptations of the classic BALANCE algorithm --- introduced first in \cite{pruhs} for the non-reusable online bipartite b-matching and later generalized to the Adwords problem in \cite{msvv} --- obtain the optimal $\left(1-1/e\right)$ competitive ratio.\footnote{Notably, \cite{full} considered the general case of stochastic usage durations and showed that BALANCE is at most $0.626$ ($<1-1/e$) competitive. They 
proposed an adaptation of BALANCE that achieves the best possible guarantee of $(1-1/e)$ for arbitrary usage distributions when capacities are large (or equivalently, when fractional allocations are allowed).} This progress leaves the open question of finding the best competitive ratio achievable by an integral online algorithm in the \omrr\, problem. In fact, no algorithm with guarantee better than 0$.5$ was known prior to our work. In this paper, we take a first step towards answering this open question by studying the following fundamental question. 

\smallskip
\begin{displayquote}
\emph{Does there exist an integral (randomized) online algorithm that can beat the naive greedy algorithm in our problem? In other words, can we obtain a constant competitive ratio $\Gamma$ strictly larger than $0.5$ for the (vertex-weighted) online bipartite matching with reusable resources?}
\end{displayquote}
\smallskip

 Our main result answers the above question in affirmative. In particular, we propose two different algorithmic approaches for designing online algorithms whose competitive ratios are \emph{strictly} better than $0.5$. In a nutshell, our first approach is based on modifying the RANKING algorithm of \cite{kvv} (and the perturbed greedy of \cite{goel}) to deal with challenges unique to the matching of reusable resource. Our second approach is based on extending the concept of online correlated selection (OCS), introduced and studied in \cite{fahrbach2020edge}, to the case of reusable resources. We refer to this extended procedure as online correlated rental (\ocs\,). We then design a primal-dual algorithm that uses this algorithmic construct as a sub-routine. Both approaches eventually lead to structured ways of correlating the randomized tie-breaking rule of the greedy algorithm across different rounds. Both approaches also rely on the primal-dual analysis framework to quantify the effect of these structures on improving the competitive ratio of the naive greedy algorithm.  By incorporating these two different approaches, we propose two final algorithms, namely \emph{Periodic Reranking} (\Cref{pr}) and \emph{OCR-based Primal-Dual} (\Cref{alg:primal-dual}). Our main technical results are the following two theorems, establishing the first competitive ratio results that beat the competitive ratio $0.5$ of the naive greedy algorithm for the (vertex-weighted) online bipartite matching of reusable resources.

\smallskip

	\begin{theorem}[Main Result I] \label{ref}
	The Periodic Reranking algorithm (Algorithm \ref{pr} with $\beta=0.89$), achieves a competitive ratio of $\Gamma\approx 0.589$ in the (vertex-weighted) \omrr\, problem.
\end{theorem}

	\begin{theorem}[Main Result II] \label{ref2}
	The OCR-based Primal-Dual algorithm (Algorithm~\ref{alg:primal-dual} with the OCR subroutine as described in \Cref{alg:ocs-sym}), achieves a competitive ratio of $\Gamma\approx 0.505$ in the (vertex-weighted) \omrr\, problem.
\end{theorem}

We leave the discussions around the details of our technical contributions including the details of our algorithmic constructs and our proof techniques to \Cref{sec:reranking} (first approach) and \Cref{sec:ocr-primal-dual} (second approach). In Appendix \ref{appx:cont}, we consider a variant of our setting where arrival times are continuous. Our main results extend as is to this setting too.

In what follows we elaborate on some further related work to help  positioning our results better in the literature.

	\subsection{Further Related Work}\label{sec:prev}

\medskip
Our work relates to and contributes to several streams of literature in operations research and computer science.

\smallskip
\noindent\emph{Online bipartite allocations.} Our results fit into the rich literature on online bipartite matching (with vertex arrival) and its extensions. Besides the work mentioned earlier, the analysis of RANKING was later clarified and considerably simplified by \cite{baum} and \cite{goel2}. Another line of work closely related to our problem studies the online assortment optimization with non-reusable resources~\citep[e.g.,][]{negin,ma2020algorithms} and reusable resources~\citep[e.g.,][]{reuse,feng,feng3,full} in the adversarial setting, in which the online algorithm decides on subsets of resources to display to an adversarial stream of arriving consumers. Several of these work introduce and analyze variants of the ``inventory balancing" algorithm --- which are all inspired by the BALANCE algorithm and its analysis in the seminal work of \cite{msvv} for the Adwords problem. In addition to these settings, there is a vast body of work on online matching and (non-reusable) resource allocation in stochastic and hybrid/mixed models of arrival. For a comprehensive review of these work,
see \cite{survey}. Last but not least, \cite{moharir2015online} introduced a unit inventory model where  
the arrival sequence is divided into slots and resources are reusable with a (deterministic) usage duration of one slot. At the beginning of each slot, arrivals are sequentially revealed within a short (infinitesimal) amount of time. The number and types of arrivals in a slot is arbitrary. 
A resource can be matched to at most one arrival in each slot and 
a resource matched in slot $t$ is available for rematch in slot $t+1$.
Consequently, 
the decision across slots are independent and it can be shown that the classic RANKING algorithm is still $(1-1/e)$ competitive. \cite{moharir2015online} proposed several new algorithms, including a $(1-1/e)$ competitive reranking algorithm that samples a new rank for the resources at the beginning of each slot. In fact, they show this result for the more general case where arrivals have heterogeneous match deadlines. Overall, their setting and results are incomparable to ours.

\smallskip
\noindent\emph{Primal-dual framework.}  Our randomized dual-fitting analysis borrows some aspects of the randomized primal-dual analysis of RANKING (and its extensions) discovered  in \cite{devanur}. An economics interpretation of this primal-dual analysis is also studied recently in \cite{eden2021economics}. At high-level, our primal-dual analysis also bears some resemblance with the primal-dual framework in the Adwords problem~\citep{buchbind,msvv}, online fractional matching with free-disposal~\citep{displayad}, the online bipartite matching with concave returns on the offline nodes~\citep{devanur2012online}, and the modified primal-dual framework for Adwords with random permutations~\citep{devhay}. Our dual-fitting style of analysis shares similarities with ~\cite{esfandiari2015online,huang2020fully2}. Another related recent line of work is on multi-stage bipartite allocations and matchings, where similar-in-spirit primal-dual analysis frameworks are discovered using ideas from convex programming duality, e.g., \cite{feng2020batching,feng2021two}. 

\smallskip
\noindent\emph{Online correlated selection.} As mentioned earlier, the concept of online correlated selection was first introduced in \cite{fahrbach2020edge} and subsequently used in \cite{huang2020adwords}. Furthermore, there are very recent work on improving this technique and extending it to the more general setting of multi-way online correlated selection~\citep{blanc2022multiway,gao2022improved}. Importantly, all of these work consider the non-reusable version of the problem while our technical developments in \Cref{sec:OCR} consider resources that are reusable.

\smallskip
\noindent\emph{Other models.} We note that a parallel stream of work considers the online bipartite allocations of reusable (and also non-reusable) resources in a Bayesian setting, both for the online prophet inequality matching~\citep[e.g.,][]{alaei,dickerson} and its extension to online Bayesian assortment optimization~\citep[e.g.,][]{RST18, baek, ma2021dynamic,feng2}. In this setting, the arrival type (which encode rewards and usage durations) are drawn from known distributions which can be varying across time. For a detailed review of these settings, see \cite{reuse} and \cite{feng2}. The model and the results in this line of work is incomparable to ours due to the fundamental difference between adversarial arrival and Bayesian arrival. 
Some other problems that are indirectly related to us are  (i) online bipartite stochastic matching~\citep[e.g.,][]{feldman,bahman,vahideh,mehta,jaillet}, (ii) online bipartite matching with stochastic rewards~\citep[e.g.,][]{deb,stochrew,huang} and (iii) the rich literature on stochastic i.i.d. online packing LP and convex programming with large budgets~\citep[e.g.,][]{FMMS-10,DJSW-11,AD-14}.


\paragraph{\textbf{Organization}} In \Cref{sec:LP}, we start by formalizing the LP relaxation  of the \omrr\ problem. We then present our first algorithm based on the idea of periodic reranking in \Cref{sec:reranking}. We provide a primal-dual analysis for this algorithm in \Cref{sec:reranking-analysis}. We then switch to our second algorithm based on the idea of online correlated selection for reusable resources in \Cref{sec:ocr-primal-dual}. We first present the idea of OCR in \Cref{sec:OCR}. We then describe our OCR-based primal-dual algorithm in \Cref{sec:primal-dual-alg} with the construction of the dual variables, and we finally show how the OCR guarantee helps in analyzing the competitive ratio using a primal-dual proof in \Cref{sec:primal-dual-analysis}.

%% file: tex/LP.tex
\section{LP Relaxation and Primal-Dual Framework}
\label{sec:LP}
 Our algorithms and their analyses rely on the linear programming relaxation of \omrr. In this LP relaxation, we associate a variable $x_{i,j}$ to every edge $(i,j) \in E$ of the graph. Since every online vertex $j$ should be matched to at most one offline vertex, the summation of $x_{i,j}$ over edges incident to $j$ should be at most $1$. On the other hand, offline vertices are reusable and we know that if an offline vertex $i$ gets matched, it will be available for re-allocation after $d$ time units. This is equivalent to saying that an offline vertex $i$ should be matched at most once in every time interval of size of $d$. For every offline vertex $i$ we have a set of constraints in the LP relaxation which states that the summation of the edges incident to $i$ in any interval of length $d$ should be at most $1$. The primal and dual linear programs of \omrr\ are given below.

\begin{figure}[h]
\begin{align*}
	\begin{array}{lll|lll}
	\multicolumn{3}{c|}{\text{Primal } (P)} & \multicolumn{3}{c}{\text{Dual } (D)} \\ \hline  
	\max & \sum_{(i,j) \in E} r_i\cdot x_{i,j} & s.t. &
	\min & \sum_{i \in V, t \in [1,|U|]} \alpha_{i,t} + \sum_{j \in U} \beta_{j} & s.t.\\
	\forall j \in U: & \sum_{i \in V} x_{i,j} \le 1 & & \forall (i,j) \in E:& \beta_j + \sum_{t=\max\{j-d+1,1\}}^{j} \alpha_{i,t} \ge r_i &\\
	\forall i \in V, t \in [1,|U|]: & \sum_{j \in [t,t+d-1]} x_{i,j} \le 1 & & \forall i \in V, t \in [1,|U|]: & \alpha_{i,t} \ge 0 & \\
	\forall i \in V,j \in U:& x_{i,j} \ge 0 &  & \forall j \in U: & \beta_j \ge 0 &\\
	\end{array} 
\end{align*}
\vspace{-0.5cm}
\end{figure}

It is interesting to mention that unlike the LP for the bipartite matching with non-reusable resources, the above LP relaxation above is not necessarily integral (and hence it is a strict relaxation to the optimum offline benchmark).
\begin{remark}
\label{lp-gap}
The integrality gap of the LP relaxation of \omrr\ is at least $\frac{7}{6}$. Consider the graph below, and assume that $d=3$. It is easy to see that the size of the optimal matching is $3$. However, we can get a feasible fractional solution of size $3.5$ by setting $x_{i,j}$ to $0.5$ for every edge.

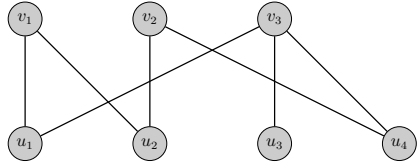
\begin{figure} [h]
    \centering
    \scalebox{.6}{\input{tex/integral-gap.tex}}
    \caption{An example for the \omrr\ with an integrality gap of $\frac{7}{6}$. In this example  $V=\{v_1, v_2, v_3\}$ and $U=\{u_1, u_2, u_3, u_4\}$.}
    \label{fig:badexample-notintegral}
\end{figure}

\end{remark}

We use the dual-fitting primal-dual framework to provide competitive ratio guarantees for our algorithms. This is a versatile and general technique for proving guarantees for online matching and related problems. Throughout the paper, our algorithms always maintain a feasible primal solution. To do so, we let $x_{i,j}$ to be the probability that edge $(i,j)$ is in the matching generated by the algorithm. Then, the primal objective $P$ is equal to the expected objective of the algorithm. We also construct a \emph{dual certificate} for each algorithm. This dual will either be constructed only for the sake of the analysis (\Cref{sec:reranking-analysis}) or it will be maintained and used by the online algorithm itself (\Cref{sec:primal-dual}). The following lemma describes how the dual certificate helps with proving competitive ratios. The proof is available in Appendix~\ref{appx:missing}.

\begin{lemma}
\label{lem:primal-dual}
Suppose an online algorithm maintains a primal assignment $\{x_{i,j}\}$ with objective value $P$, and there is a dual assignment $\{\alpha_{i,t},\beta_j\}$ with objective value $D$ satisfying these conditions:
\begin{enumerate} [label=(\roman*)]
\smallskip
    \item \textbf{\textit{Reverse weak duality:}} $P \ge D$,
     \item \textbf{\textit{Approximate dual feasibility:}} there exists $\Gamma\in(0,1]$, so that
    $$\forall (i,j)\in E:~~\beta_j + \sum_{t=\max\{j-d_i+1,1\}}^{j} \alpha_{i,t} \ge \Gamma r_i.$$ 
\end{enumerate}
Then, the algorithm is $\Gamma$-competitive.
\end{lemma}

%% file: tex/integral-gap.tex
\begin{tikzpicture}
    \node[main node] (v_1) {$v_1$};
    \node[main node] (v_2) [right = 2cm  of v_1]  {$v_2$};
    \node[main node] (v_3) [right = 2cm  of v_2]  {$v_3$};
    \node[main node] (u_1) [below = 2cm  of v_1] {$u_1$};
    \node[main node] (u_2) [right = 2cm  of u_1] {$u_2$};
    \node[main node] (u_3) [right = 2cm  of u_2] {$u_3$};
    \node[main node] (u_4) [right = 2cm  of u_3] {$u_4$};

    \path[draw,thick]
    (v_1) edge node {} (u_1)
    (v_1) edge node {} (u_2)
    (v_2) edge node {} (u_2)
    (v_2) edge node {} (u_4)
    (v_3) edge node {} (u_1)
    (v_3) edge node {} (u_3)
    (v_3) edge node {} (u_4)
    
    ;
\end{tikzpicture}

%% file: tex/reranking.tex
	\section{First Approach: Periodic Reranking }
	\label{sec:reranking}
	In this section, we provide the details of our first approach based on a natural extension of the classic Ranking algorithm of \cite{kvv}, and propose the \emph{Periodic Reranking (PR)} algorithm (\Cref{pr}). In \Cref{sec:reranking-analysis}, we establish a lower bound on the competitive ratio of this algorithm using a primal-dual analysis.
	
	The main new ideas in our algorithm are (i) using reranking, i.e., drawing fresh randomized ranks occasionally, and (ii) reranking on a periodic schedule every $d$ time units, where $d$ is the usage duration of resources. To gain insight into the usefulness of reranking and periodicity, consider the classic RANKING algorithm on the example below. 
	\begin{eg}\label{ranker}
		\emph{	Consider a setting with two reusable resources $\{1,2\}$, identical rewards, usage duration $d>0$ and four arrivals. The first and third arrivals have edges to both resources. The second arrival only has an edge to resource 2. The fourth arrival only has an edge to resource 1. The first two arrivals occur in close proximity to each other (less than $d$ time apart). The second and third arrivals are well separated in time (more than $d$ units apart). Finally, the last two arrivals also occur close to each other (similar to the first two). 
			Observe that the matching decisions at arrivals one and two have no impact on the availability of resources at arrival three. The RANKING algorithm will randomly rank the two resources. Since ranks are not changed, arrivals one and three are always matched to the same resource. Therefore, arrival two is matched if and only if arrival four is not matched. In contrast, the optimal match is obtained by ranking resource 1 over resource 2 for the first two arrivals and then \emph{reversing the ranks} for the remaining two arrivals. 
	} \end{eg}  
	
	The example does not give an upper bound on the overall performance of RANKING but illustrates a key difficulty in analyzing the performance of RANKING for reusable resources. In general, when the ranking is fixed, 
	``right" matching decisions on early arrivals (matching the first arrival to resource 1), may imply 
	``wrong" decisions on later arrivals (matching the third arrival to resource 1).  Reranking provides a natural way to mitigate this analytical issue, as it untangles the the dependence between matching decisions for arrivals that are well separated across time and makes it tractable to analyze the resulting algorithm.
	At the same time, it is important to note that reranking should be done at an appropriate frequency. Consider the algorithm that reranks resources at every arrival. When vertex weights are identical, say $r_i=1\,\, \forall i\in V$, this algorithm is equivalent to the following randomized algorithm: Match every arrival (that can be matched) by sampling a resource uniformly randomly. This algorithm, called Random, is known to have worst case performance same as greedy even for non-reusable resources \citep{kvv}. 
	
	Now consider periodic reranking every $d$ units of time. First, it operationalizes the insight that a decision to match a resource to arrival $j$ does not affect the resource availability after arrival $j+d-1$. Second, within a period, this maintains the same rank and avoids the pitfall of frequent reranking.\footnote{The key insight is that in the time span of one usage duration, each resource can be matched at most once, presenting a scenario similar to non-reusable resources. Indeed, \alg\ reduces to RANKING when $d\geq T$.} We formalize the periodic reranking idea in \Cref{pr}.
	
	\smallskip


	
	\SetInd{0.5em}{0.5em}
	\begin{algorithm}[htb]
		\SetAlgoNoLine
		\textbf{Inputs:} Set of resources $V$, usage duration $d$, parameter $\beta$\; 
		Let $g(x)=e^{\beta(x-1)}$ and $S=V$\;
		\smallskip
		\textbf{Every $d$ time units:} Generate new i.i.d. ranks $y_i\sim U[0,1]\,\, \forall i\in V$\;
		\smallskip
		\For{\text{every new arrival } $j$}{
			Update set $S$ by adding resources that returned since arrival $j-1$\;
			Match $j$ to $i^*=\underset{ i\in S,\, (i,j)\in E}{\arg\max}\quad r_i (1-g(y_i))$\;
			$S=S\backslash \{i\}$\;
		}	
		\caption{Periodic Reranking (PR)}
		\label{pr}
	\end{algorithm}
	\smallskip
	
	At the start of the planning horizon, the PR algorithm (independently) samples a random seed $y_i\in U[0,1]$, for every $i\in V$. Using this seed, and a monotonically increasing trade-off function $g:\mathbb{R}\to [0,1]$, the algorithm evaluates \emph{reduced prices} $r_i(1-g(y_i))\,\,\, \forall i\in V$. Observe that the reduced prices change over time. In particular, after every $d$ units of time, \alg\ samples 
	new seeds for the resources. 
	Re-sampling over periods of length $d$ ensures that resources have a new seed every time they return to the system after a match. Given the reduced prices, \alg\ matches each arrival to an available neighbor with the highest reduced price at the moment of arrival. The name Periodic Reranking comes from the following observation. When rewards $r_i=1\,\, \forall i\in V$, due to the monotonicity of $g$, the algorithm is equivalent to reranking resources after every $d$ units of time and matching arrivals to the best ranked available neighbor.
	
	When resources are non-reusable, say $d=\lvert U \rvert$, \alg\ reduces to the Perturbed Greedy (PG) algorithm. 
	For the PG algorithm, \cite{goel} showed that choosing $g(x)=e^{x-1}$ leads to the best possible guarantee of $(1-1/e)$ for OBM with arbitrary rewards. In \alg, we consider the family of functions $g(x)=e^{\beta(x-1)}$ parameterized by $\beta>0$. Our analysis dictates the choice of $\beta$. In particular, $\beta =0.89$ optimizes the guarantee that can be achieved with our analysis.

%% file: tex/reranking-analysis.tex
	\subsection{Competitive ratio analysis using primal-dual}
	\label{sec:reranking-analysis}
	Our analysis relies on the randomized primal-dual framework of \cite{devanur}, which is based on constructing a randomized dual certificate. Before getting into the details of our construction, we define some required notations.
	
	\smallskip
	\noindent \textbf{Notations:} Recall that PR reranks in fixed periods of length $d$. Let $K=\ceil{|U|/d}$ denote the total number of periods and let $k(j)$ denote the period that contains arrival $j\in U$. 
	To ensure that $k(j)-1$ is well defined for every $j\in U$, we add a dummy period (of time interval $d$) prior to the first arrival. This period does not have any arrivals and simply ensures that $k(j)\geq 2$ for every arrival. Let $y^k_i$ denote the $k$-th seed of resource $i$. Note that $y^k_i$ is the seed of $i$ in period $k\in[K]$. 
	Let $Y$ denote the vector of all random seeds. Given a resource $i\in V$ and arrival $j\in U$, let $Y_{-i,j}$ denote the vector of all seeds except $y^{k(j)}_i$ and $y^{k(j)-1}_i$. In other words, $Y_{-i,j}$ captures all seeds except the seed of $i$ during periods $k(j)-1$ and $k(j)$. We use $E_{y^{k(j)}_i,\, y^{k(j)-1}_i}[\cdot]$ to denote expectation with respect to the randomness in seeds $y^{k(j)}_i$ and $y^{k(j)-1}_i$.

	
	
	\noindent \textbf{Dual construction:} 
	In order to define our dual candidate, we first define random variables $\beta_j(Y)$, $\alpha_{i,j}(Y)$ and subsequently set $\beta_j=E_{Y}[\beta_j(Y)]$ and $\alpha_{i,j}=E_{Y}[\alpha_{i,j}(Y)]$. 
	Inspired by \cite{devanur}, we set $\beta_j(Y)$ and $\alpha_{i,j}(Y)$ as follows. 
\begin{itemize}
    \item Initialize all dual variables to 0.
    \item Conditioned on $Y$, for each match $(i,j)$ in PR set of assigned edge, let  
	\begin{eqnarray}
		\beta_j(Y)&= &r_i\left(1-g\left(y^{k(j)}_i\right)\right),\label{lambda}\\
			\alpha_{i,j}(Y)&= & r_i\, g\left(y^{k(j)}_i\right). \label{theta}
	\end{eqnarray}
\end{itemize}

	
		\noindent \textbf{Dual certification:} Now we show our constructed dual satisfies \Cref{lem:primal-dual} for $\Gamma\approx 0.5893$.
	\begin{lemma}
		The dual candidate given by \eqref{lambda} and \eqref{theta} satisfies constraint (i) of \Cref{lem:primal-dual}.
	\end{lemma}
	\begin{proof}{Proof.}
		Let $\alg(Y)$ denote the matching output by PR given seed vector $Y$. From \eqref{lambda} and \eqref{theta}, we have,
		\[\beta_j(Y)+\alpha_{i,j}(Y)=r_i\quad \forall (i,j)\in \alg(Y).\]
		Summing over all edges in the matching $\alg(Y)$ completes the proof.
	
		\hfill\Halmos
	\end{proof}
	
	It remains to show that constraints $(ii)$ of \Cref{lem:primal-dual} hold for the desired value of $\Gamma$. For OBM, \cite{devanur} prove a stronger statement in terms of conditional expectations. We follow a similar strategy.
	
	\begin{lemma}
		Consider an edge $(i,j)\in E$ and seed $Y_{-i,j}$. Suppose that for the candidate solution given by \eqref{lambda} and \eqref{theta}, we have 
		\begin{equation}\label{dual2c}
			E_{y^{k(j)}_i,\, y^{k(j)-1}_i}\left[\beta_j(Y) + \sum_{t=\max\{j-d+1,1\}}^{j} 
			\alpha_{i,t}(Y) \,\big|\, Y_{-i,t}\right]\geq \Gamma\, r_i,
		\end{equation}
		
		for some value $\Gamma>0$. Then, constraint (ii) of \Cref{lem:primal-dual} is satisfied for edge $(i,j)\in E$  with the same $\Gamma$. 
	\end{lemma}
	\begin{proof} {Proof.}
		The lemma follows by taking expectation over $Y_{-i,j}$ on both sides of \eqref{dual2c}. 
		
		\hfill\Halmos
	\end{proof}
	
	For a given value of seed $y^{k(j)-1}_i$, resource $i$ may be matched to an arrival in period $k(j)-1$ such that it is unavailable at $j$ for all values of seed $y^{k(j)}_i$.  
	Therefore, a stronger version of \eqref{dual2c} where we also fix $y_i^{k(j)-1}$ and consider a conditional expectation only w.r.t.\ random seed $y^{k(j)}_i$, does not hold for any non-trivial value of $\Gamma$. This necessitates an analysis where consider an expectation w.r.t.\ both $y^{k(j)-1}_i$ and $y^{k(j)}_i$. 
	In a hypothetical scenario where $i$ is available at $j$ for all values of $y^{k(j)-1}_i$, it can be shown that inequality \eqref{dual2c} holds for $\Gamma=(1-1/e)$ (when $\beta=1$). Of course, in reality, $i$ may not be available at $j$ for some values of $y^{k(j)-1}_i$ 
	and this scenario leads to a lower value of $\Gamma$ (= 0.589) in our analysis.  

	We now prove \eqref{dual2c} for every edge $(i,j)\in E$ and seed $Y_{-i,j}$. To this end, fix an arbitrary edge $(i,j)$ and seed $Y_{-i,j}$.  To simplify notation, let $y^1=y^{k(j)-1}_i$ and $y^2=y^{k(j)}_i$. Further, let $\alg(y^1, y^2)$ denote the matching output by \alg\ given seeds $y^1,y^2$ for $i$ and with other seeds fixed according to $Y_{-i,j}$. Since $Y_{-i,j}$ is fixed, 
	for simplicity,
	let $\beta_j(y^1,y^2)$ denote $\beta_j(y^1,y^2,Y_{-i,j})$. Similarly, for $t\in U$, let $\alpha_{i,t}(y^1,y^2)$ denote $\alpha_{i,t}(y^1,y^2,Y_{-i,j})$. We also write the conditional expectation $E_{y^1,y^2}[\cdot \mid Y_{-i,j}]$ as $E_{y^1,y^2}[\cdot]$. 
	
	Let $S_t(y^1,y^2)$ denote the set of resources available at arrival $t\in U$ 
	in $\alg(y^1,y^2)$.
	Given $y^1\in[0,1]$, for every arrival $t \in U$, 
	define the critical threshold $y^c_t(y^1)$ as the solution to,
	\[r_i\left(1-g\left(y^c_t(y^1)\right)\right) = \max_{v\in S_t(y^1,1),\, (v,t)\in E}r_v\left(1-g\left(y^{k(t)}_v\right)\right). 	\]
	Due to the monotonicity of function $g$, there is at most one solution to this equation. If there is no solution, we let $y^c_t(y^1)=0$. At a high level, the critical threshold at arrival $t$ captures the highest reduced price at the arrival when resource $i$ is ``removed" in period $k(j)$ (achieved by setting $y^2=1$). Similar to the analysis of OBM (\cite{devanur}), this scenario serves as a foundation for establishing lower bounds on $E_{y^{1},\, y^{2}}\left[\beta_j(y^1,y^2)\right]$ and $E_{y^{1},\, y^{2}}\left[ \sum_{t=\max\{j-d+1,1\}}^{j} 
	\alpha_{i,t}(y^1,y^2)\right]$. 
	
	Recall that $k(j)$ is the period that contains $j$. 
	Let $p(j)=\{1,\cdots,j\}\cap k(j)$ denote the sub-interval of period $k(j)$ that includes all arrivals prior to (and including) $j$. We let $S_{p(j)}(y^1,1)=\cup_{t \in p(j)} S_t(y^1,1)$ i.e., $S_{p(j)}(y^1,1)$ denotes the set of all resources that are available at some point of time in interval $p(j)$. The next lemma gives useful lower bounds on 
	$\beta_j(y^1,y^2)$ and $\sum_{t=\max\{j-d+1,1\}}^{j} \alpha_{i,t}(y^1,y^2)$, when resource $i$ is available at some point in the interval $p(j)$.
	
	\begin{lemma}\label{devlb}
		Given $y^1\in[0,1]$ such that $i\in S_{p(j)}(y^1,1)$, we have, 
		\begin{enumerate}[label=\alph*)]
			\item $\beta_j(y^1, y^2)\,\geq\,\beta_j(y^1,1)\, \geq\, r_i \left(1-g(y^c_j(y^1))\right)\quad \forall y^2\in[0,1]$.
			\item $\sum_{t=\max\{j-d+1,1\}}^{j} 
			\alpha_{i,t}(y^1,y^2) \, \geq\, \onee(y^2<y^c_j(y^1))\, r_i g(y^2)\quad \forall y^2\in[0,1].$
		\end{enumerate}
	\end{lemma}
	Since every resource is matched at most once within each period, the bounds in Lemma \ref{devlb} are quite similar to their counterparts in the classic OBM setting where resources are matched at most once \citep{devanur}. For a proof, see Appendix \ref{appx:missing}. 
	
	Next, we show a useful lower bound on 
	$\sum_{t=\max\{j-d+1,1\}}^{j} \alpha_{i,t}(y^1,y^2)$ when $i\not\in S_{p(j)}(y^1,1)$. %
	
	\begin{lemma}\label{unavail}
		Given $y^1\in[0,1]$ such that $i\not\in S_{p(j)}(y^1,1)$, we have, 
		$\sum_{t=\max\{j-d+1,1\}}^{j} \alpha_{i,t}(y^1,y^2)\geq  r_i g(y^1)\,\, \forall y^2\in[0,1]$.
	\end{lemma}
	\begin{proof}{Proof.}
		Given $i\not\in S_{p(j)}(y^1,1)$, we have that $i$ 
		is matched in period $k(j)-1$ to an arrival $t'>j-d$. 
		From \eqref{theta}, 
		we have $\alpha_{i,t'}(y^1, y^2)=r_i\, g(y^1)$. Thus, 
		\[\sum_{t=\max\{j-d+1,1\}}^{j} \alpha_{i,t}(y^1,y^2)\,=\, \alpha_{i,t'}(y^1, y^2) \, \geq\, r_i\, g(y^1)\quad \forall y^2\in[0,1].\]
		\hfill\Halmos
	\end{proof}
	
	Notice that Lemma \ref{devlb} and Lemma \ref{unavail} apply to mutually exclusive and exhaustive scenarios. 
	Combining Lemma \ref{devlb}$(b)$ with Lemma \ref{unavail} gives a lower bound on $\sum_{t=\max\{j-d+1,1\}}^{j} \alpha_{i,t}(y^1,y^2)$ for all $(y^1,y^2)\in[0,1]^2$. In the proof of Lemma \ref{combine}, we turn this into a desired lower bound on the expectation $E_{y^1,y^2}\left[\sum_{t=\max\{j-d+1,1\}}^{j} \alpha_{i,t}(y^1,y^2)\right]$. 
	
	In the scenario where $i\in S_{p(j)}(y^1,1)$, Lemma \ref{devlb}$(a)$ lower bounds $\beta_j(y^1,y^2)$ as a function of the critical threshold $y^c_j(y^1)$. It remains to lower bound $\beta_j(y^1, y^2)$ when $i\not \in S_{p(j)}(y^1,1)$ and find convenient bounds on $y^c_j(y^1)$. To this end, Lemma \ref{y1} first gives a sharp characterization of the set of values of $y^1$ that lead to each scenario. Lemma \ref{connect} builds on this characterization to upper bound $y^c_j(y^1)$ in two crucial scenarios. Finally, Lemma \ref{combine} fills in the gaps and puts the various pieces together. The next lemma gives a structural result that will be used to prove Lemma \ref{y1}. 
	\begin{lemma}\label{early}
		Consider a value $z\in[0,1]$ such that for $y^1=z$, $i$ is matched to some arrival, say $t(z)$, in period $k(j)-1$. Then, for every $y^1\leq z$, $i$ is matched in period $k(j)-1$ to arrival $t(z)$ or an arrival that precedes it.
		
	\end{lemma}
	\begin{proof}{Proof.}
		Recall that except $y^1$ and $y^2$, all seeds are fixed. The value of $y^2$ does not affect the output of \alg\ in periods prior to $k(j)$. Similarly, the value of $y^1$ does not affect the matching prior to period $k(j)-1$. 
		Since every resource can be matched at most once during a single period, when $y^1=z$, $t(z)$ is the unique arrival matched to $i$ during period $k(j)-1$. 
		
		Now, let $r_i(y^1)=r_i(1-g(y^1))$ and consider the change in the matching during period $k(j)-1$ as we vary $y^1$ in the interval $(0,z)$. 
		Suppose there exists a value $y^1 = z'$, with $z'<z$, such that $i$ is not matched prior to $t(z)$ in period $k(j)-1$ (if no such value exists, we are done). Then, for $y^1=z'$, $i$ is available at $t(z)$ and the matching prior to $t(z)$ is identical to the matching when $y^1=z$. Hence, the set of resources available at $t(z)$ is identical for both values of $y^1$. Since $r_i(z')\geq r_i(z)$ (by monotonicity of function $g$ for $\beta>0$), $i$ must be matched to $t(z)$ when $y^1=z'$. This completes the proof.
		\hfill\Halmos
	\end{proof}
	\begin{lemma}\label{y1}
		There exists values $z_1, z_2 \in [0,1]$, such that $z_1\leq z_2$ and, 
		\begin{enumerate}[label=\alph*)]
			\item $i\in S_{p(j)}(y^1,1) 
			\quad \forall y^1 \in (z_2,1]$ and $i$ is not matched to any arrival in period $k(j)-1$. 
			\item $i\not\in S_{p(j)}(y^1,1)
			\quad \forall y^1\in(z_1,z_2)$ and $i$ is matched to some arrival in period $k(j)-1$.  
			\item $i\in S_{p(j)}(y^1,1) 
			\quad \forall y^1 \in [0,z_1)$ and $i$ is matched to some arrival in period $k(j)-1$. 
		\end{enumerate}
	\end{lemma}
	\begin{proof}{Proof.}
		Recall that $S_{p(j)}(y^1,1)$ denotes the set of all resources that are available at some point of time in interval $p(j)$. Observe that the value of $y^2$ does not influence the scenario i.e., whether $i$ is in (or not in) $S_{p(j)}(y^1,1)$. 
		
		Let $z_2\in[0,1]$ be the highest value 
		such that 
		for $y^1=z_2$, $i$ is matched in period $k(j)-1$. Set $z_2=0$ if no such value exists. 
		From Lemma \ref{early}, we have that for every $y^1<z_2$, $i$ will continue to be matched in period $k(j)-1$ and, in fact, to (possibly) earlier arrivals. Thus, there is a unique threshold $z_2\in(0,1)$ such that $i$ is matched in period $k(j)-1$ for every $y^1\leq z_2$, and unmatched in period $k(j)-1$ for every $y^1>z_2$. If $i$ is unmatched in period $k(j)-1$, then $i\in S_{p(j)}(y^1,1)$. This gives us part $(a)$ of the lemma. 
		
		Next, let $z_1\in[0,z_2]$ be the highest value such that $i\in S_{p(j)}(y^1,1)$ for for $y^1=z_1$. In other words, $i$ returns from its match in $k(j)-1$ in time to be available at some arrival in $p(j)$. Set $z_1=0$ if no such value exists. From Lemma \ref{early}, for every $y^1< z_1$, $i$ is matched (possibly) even earlier in period $k(j)-1$. Therefore, $i\in S_{p(j)}(y^1,1)$ for every $y^1< z_1$. This corresponds to part $(c)$ of the lemma. 
		
		Finally, by definitions of thresholds $z^1$ and $z^2$, when $y^1\in(z_1,z_2)$, we have that $y^1$ is matched in period $k(j)-1$ but $i\not\in S_{p(j)}(y^1,1)$. This corresponds to part $(b)$. 
		\hfill\Halmos\end{proof}
	\begin{lemma}\label{connect} The following statements are true.
		\begin{enumerate}[label=\alph*)]
			\item For every $y^1\in(z_2,1]$, 
			we have $y^c_j(y^1)= y^c_j(1)$.
			\item For every $y^1\in(z_1,z_2)$, 
			we have $y^c_j(y^1)\leq y^c_j(1)$.
		\end{enumerate}
	\end{lemma}
	\begin{proof}{Proof.}
		From Lemma \ref{y1}$(a)$, we have that for every $y^1\in(z_2,1]$, resource $i$ is unmatched in period $k(j)-1$. Therefore, with $y^2$ fixed at 1, the matching output by \alg\ is identical for every value of $y^1>z_2$. This proves part $(a)$.
		
		Let $r_t(y^1,y^2)$ denote the reduced price of the resource matched to arrival $t\in U$ in the matching $\alg(y^1,y^2)$. Set $r_t(y^1,y^2)=0$ if $t$ is unmatched. To prove part $(b)$, fix an arbitrary value $y^1=z\in(z_1,z_2)$ and consider the matching $\alg(z,1)$. From Lemma \ref{y1}$(b)$, we have that $i$ is matched in period $k(j)-1$ but does not return prior to $j$. Let $t(z)$ denote the arrival matched to $i$ in period $k(j)-1$ 
		and let $T'=\{t \in U \mid t(z)\,\leq\, t\,\leq\, j\}$. Since $t(z)\geq j-d+1$, every resource is matched to at most one arrival in $T'$. 
		Now, given that $g(x)$ is strictly increasing in $x$ (for $\beta>0$), to prove $y^c_j(z)\leq y^c_j(1)$, it suffices to show that 
		$r_t(y^1,1)\geq r_t(1,1)\,\,\, \forall\, y^1\in(z_1,z_2),\,\, t\in T'$. Note that $1-g(1)=0$ for every $\beta$. Therefore, when $y^1=1$, the reduced price of arrival matched to $i$ is 0, same as if the arrival were unmatched. Combining this observation with the fact that \alg\ matches each arrival greedily based on reduced prices, the inequality $r_t(y^1,1)\geq r_t(1,1)$ follows from, 
		\begin{equation}\label{nest}
			S_t(1,1)\backslash\{i\}\subseteq S_t(y^1,1)\quad \forall y^1\in(z_1,z_2),\, t\in T'.
		\end{equation}
		
		We prove \eqref{nest} via induction over the set $T'$. The first arrival in $T'$ is $t(z)$. From Lemma \ref{early}, prior to $t(z)$, resource $i$ is not matched to any arrival in period $k(j)-1$ in the matching $\alg(1,1)$. Thus, $\alg(1,1)$ and $\alg(z,1)$ are identical prior to $t(z)$ and $S_{t(z)}(1,1)= S_{t(z)}(z,1)$. Now, suppose that \eqref{nest} holds for all arrivals $t< t'$, for some $t'\in T'$. We show that \eqref{nest} holds for arrival $t'$ as well.
		
		For the sake of contradiction, suppose there exists a resource $v\in S_{t'}(1,1)\backslash (S_{t'}(z,1)\cup \{i\})$. 
		Recall that $S_{t}(1,1)\backslash\{i\}\subseteq S_t(z,1)$ for all $t<t'$. Thus, $v$ is matched to arrival $t'-1$ in $\alg(z,1)$, where $t'-1>t(z)$. Since every resource is matched to at most one arrival in $T'$, we have, 
		$S_{t'}(1,1)\subseteq S_{t'-1}(1,1)$. Thus, $v\in S_{t'-1}(1,1)\backslash\{i\}\subseteq S_{t'-1}(z,1)$ i.e., in $\alg(1,1)$, resource $v$ is available but not matched to $t'-1$. This contradicts the fact that \alg\ matches greedily based on reduced prices. 
		
		\hfill\Halmos
		
	\end{proof}

	The next lemma combines all possible scenarios given in Lemma \ref{y1} to lower bound \eqref{dual2c}. Then, the proof of Theorem \ref{ref} follows via standard algebraic arguments. Let $g(x)=e^{\beta(x-1)}$ for some $\beta\in(0,1]$. Let $G(x)$ be the antiderivative of $g(x)$.
	
	\begin{lemma}\label{combine}
		There exists values $z_1, z_2\in[0,1]$ with $z_1\leq z_2$, such that
		\begin{eqnarray}
			&&E_{y^1, y^2}\left[\beta_j(y^1,y^2) + \sum_{t=\max\{j-d+1,1\}}^{j} 
			\alpha_{i,t}(y^1,y^2)\right]\nonumber\\
			&&\quad \geq  r_i\, \Bigg[G(z_2)-G(z_1)+ (1-g(y^c_t(1)))(1-z_1)+(1-z_2)\left(G(y^c_t(1))-G(0)\right)
			+z_1(1-g(0))
			\Bigg].\nonumber
		\end{eqnarray}
		
	\end{lemma}
	\begin{proof}{Proof.} Observe that for any random variable $X$ derived from $y^1, y^2$, we have,
		\[E_{y^1,y^2}[X] = (1-z_2)\, E_{y^1,y^2}\left[X\mid y^1>z_2\right]+(z_2-z_1)\, E_{y^1,y^2}\left[X\mid y^1\in(z_1,z_2)\right]+z_1\, E_{y^1,y^2}\left[X\mid y^1<z_1\right].\]
		We prove the main claim by establishing lower bounds on each of the three terms on the RHS. 
		\smallskip
		
		\noindent \textbf{Case I:} $\mb{y^1>z_2}$.
		From Lemma \ref{y1}$(a)$, when $y^1>z_2$, $i$ is not matched in period $k(j)-1$. Thus, $\alg(y^1,y^2)=\alg(1,y^2)\,\,\,\forall y^1>z_2$ i.e., in this case, the matching does not change with $y^1$. From Lemma \ref{devlb}$(a)$ and Lemma \ref{connect}$(a)$, we have
		\[\beta_j(y^1,y_2)\geq r_i\,\left(1-g(y^c_j(1))\right)\quad \forall y^1\in(z_2,1],\, y^2\in[0,1].\]
		Taking expectation over randomness in $y^1, y^2$, we have 
		\[E_{y^1}\left[E_{y^2}[\beta_j(y^1,y^2) \mid y^1> z_2]\right]\geq E_{y^1}\left[r_i\, \left(1-g(y^c_j(1))\right)\right] = r_i\,\left(1-g(y^c_j(1))\right).\]
		Finally, from Lemma \ref{devlb}$(b)$ and Lemma \ref{connect}$(a)$, we have,
		\begin{eqnarray*}
			E_{y^1}\left[E_{y^2}\left[\sum_{t=\max\{j-d+1,1\}}^{j}\alpha_{i,t}(y^1,y^2)\mid y^1>z_2\right]\right] &\geq &E_{y^1}\left[E_{y^2}[\onee(y^2<y^c_j(1))\, r_i g(y^2)]\mid y^1>z_2\right],\\
			&= &E_{y^1}\left[ r_i\, \int_0^{y^c_j(1)} g(x) dx\mid y^1>z_2\right],\\
			&= &r_i\,(G(y^c_j(1))-G(0)). 
		\end{eqnarray*}
		
		\noindent \textbf{Case II:} $\mb{z_1<y^1<z_2}$. In this case, $i$ is not available in period $k(j)$ prior to arrival $t$ and the value of $y^2$ does not affect the matching until after arrival $t$. From Lemma \ref{unavail}, we have,
		\[E_{y^1,y^2}\left[\sum_{t=\max\{j-d+1,1\}}^{j} \alpha_{i,t}(y^1,y^2)\mid y^1\in(z_1,z_2)\right]\geq  r_i \int_{z_1}^{z_2} g(x)dx\,=\, r_i\, (G(z_2)-G(z_1)).\]
		From Lemma \ref{connect}$(b)$, we have
		\[E_{y^1,y^2}\left[\beta_j(y^1,y^2)\mid y^1\in(z_1,z_2)\right]\,=\,E_{y^1}\left[\beta_j(y^1,1)\mid y^1\in(z_1,z_2)\right]\,\geq\, r_i\, \left(1-g(y^c_j(1))\right). \]
		
		\noindent \textbf{Case III:} $\mb{y^1<z_1}$. 
		In this case, $i\in S_{p(j)}(y^1,1)$. From Lemma \ref{devlb}$(a)$, we have
		\begin{equation}\label{c3l}
			E_{y^2}[\beta_j(y^1,y^2)\mid y^1<z_1]\geq r_i\, \left(1-g(y^c_j(y^1))\right).
		\end{equation}
		From part $(b)$ of Lemma \ref{devlb}, 
		\begin{equation}\label{c3t}
			E_{y^2}\left[\sum_{t=\max\{j-d+1,1\}}^{j}\alpha_{i,t}(y^1,y^2)\mid y^1<z_1\right] \geq E_{y^2}\left[\onee(y^2<y^c_j(y^1))\, r_i g(y^2)\mid y^1<z_1\right]\,=\, r_i\,\left(G(y^c_j(y^1))-G(0)\right).
		\end{equation}
		Combining \eqref{c3l} and \eqref{c3t}, we have,
		\begin{eqnarray*}
			E_{y^2}\left[\beta_j(y^1,y^2)+\sum_{t=\max\{j-d+1,1\}}^{j}\alpha_{i,t}(y^1,y^2)\mid y^1<z_1\right] &\geq &r_i\left[1-g(y^c_j(y^1))+G(y^c_j(y^1))-G(0)\right],\\
			&= &r_i\left[1-g(y^c_j(y^1))+\frac{1}{\beta}\left(g(y^c_j(y^1))-g(0)\right)\right],\\
			&\geq &r_i \min_{x\in[0,1]}\left[1-g(x)+\frac{1}{\beta}\left(g(x)-g(0)\right)\right],\\
			&= &r_i\, (1-g(0)).
		\end{eqnarray*}
		The first equality uses the fact that $G(x)=\frac{1}{\beta}g(x) + c$, where $c$ is some constant. The second equality follows from the fact that $h(x)=1-g(x)+\frac{1}{\beta}\left(g(x)-g(0)\right)$, is a non-decreasing function of $x$ for $\beta\in(0,1]$. Thus,
		
		\[E_{y^1}\left[	E_{y^2}\left[\beta_j(y^1,y^2)+\sum_{t=\max\{j-d+1,1\}}^{j}\alpha_{i,t}(y^1,y^2)\mid y^1<z_1\right]\right]\geq  r_i\, (1-g(0)). \]
		\smallskip
		
		\hfill\Halmos
		\end{proof}
	
	\begin{proof}{Proof of Theorem \ref{ref}.} 	Let 
		\[f(z_1,z_2,x)=G(z_2)-G(z_1)+ (1-g(x))(1-z_1)+(1-z_2)\left(G(x)-G(0)\right)	+z_1(1-g(0)).\]
		We show that, $\min_{0\leq z_1\leq z_2\leq 1,\, x\in[0,1]} f(z_1,z_2,x)> 0.589$ for $\beta=0.89$. Then, using Lemma \ref{combine} completes the proof. First, using the fact that $G(x)=\frac{1}{\beta}g(x)+c$, where $c$ is some constant, we have,
		\begin{eqnarray}
			f(z_1,z_2,x)
			&= &\frac{1}{\beta}\left(g(z_2)-g(z_1)\right)+ (1-z_1)\left(1-g(x)\right)+\frac{1-z_2}{\beta}\left(g(x)-g(0)\right)+z_1\left(1-g(0)\right)\nonumber\\
			&= &\frac{1}{\beta}\left(g(z_2)-g(z_1)\right)+1-\frac{g(0)}{\beta}\left(1-z_2+\beta z_1\right)+\frac{g(x)}{\beta}\left( 1-z_2+\beta z_1-\beta\right)\label{exp1}
		\end{eqnarray}
		
		To find the minimum of this function, consider the following cases.
		\smallskip
		
		\noindent \textbf{Case I: $\mb{1-z_2\geq \beta(1-z_1)}$.} In this case, \eqref{exp1} is minimized at $x=0$. Thus, 
		\begin{eqnarray*}
			f(z_1,z_2,x)&\geq &\frac{1}{\beta}\left(g(z_2)-g(z_1)\right)+1-g(0)\\
			&\geq & 1-g(0)\, =\, 1-e^{-\beta}, 
		\end{eqnarray*}
		where we used the fact that $g(z_2)\geq g(z_1)$ for $\beta\geq 0$ and $z_2\geq z_1$. 
		\smallskip
		
		\noindent \textbf{Case II: $\mb{1-z_2< \beta(1-z_1)}$.} In this case \eqref{exp1} is minimized at $x=1$. Thus,
		\begin{eqnarray*}
			f(z_1,z_2,x)&\geq &\frac{1}{\beta}\left(g(z_2)-g(z_1)\right) +\frac{1-g(0)}{\beta}\left( 1-z_2+\beta z_1\right)\label{exp2}
		\end{eqnarray*}
		Observe that the function $-e^{\beta z_1}+c\, z_1$, where $c$ is some constant, is concave in $z_1$. Thus, \eqref{exp2} is minimized at $z_1=0$ or $z_1=\min\left\{z_2,1-\frac{1-z_2}{\beta}\right\}$. In fact, $z_2\geq 1-\frac{1-z_2}{\beta}$ for every $\beta\leq 1$. Therefore,
		\begin{eqnarray*}
			f(z_1,z_2,x)&\geq &\min\left\{\frac{1}{\beta}\left(e^{\beta(z_2-1)}-e^{-\beta}\right) +\frac{1-e^{-\beta}}{\beta}\left( 1-z_2\right),\,\, \frac{1}{\beta}\left(e^{\beta(z_2-1)}-e^{z_2-1}\right) +1-e^{-\beta}\right\}\\
			&\geq &\min\left\{\frac{1}{\beta}\left(e^{\beta(z_2-1)}-e^{-\beta}\right) +\frac{1-e^{-\beta}}{\beta}\left( 1-z_2\right),\,\, 1-e^{-\beta}\right\},
		\end{eqnarray*}
		where we used the fact that $e^{\beta(z_2-1)}-e^{z_2-1}\geq 0$ for $\beta\leq 1$ and $z_2\in[0,1]$. 
		
		Combining both cases, we have that
		\[	f(z_1,z_2,x)\geq \min\left\{\min_{z_2\in[0,1]}\frac{1}{\beta}\left(e^{\beta(z_2-1)}-e^{-\beta}\right) +\frac{1-e^{-\beta}}{\beta}\left( 1-z_2\right),\,\, 1-e^{-\beta}\right\}.\]
		The first term inside the minimum is the solution to a convex minimization problem. It is easy to (numerically) verify that when $\beta=0.89$, both terms are greater than $0.5893$, giving us the desired guarantee. We numerically tried different values of $\beta$ to arrive at the conclusion that $0.89$ is the best choice for $\beta$. 
		\begin{figure}[htbp]\label{plot}
			\includegraphics[scale=0.4]{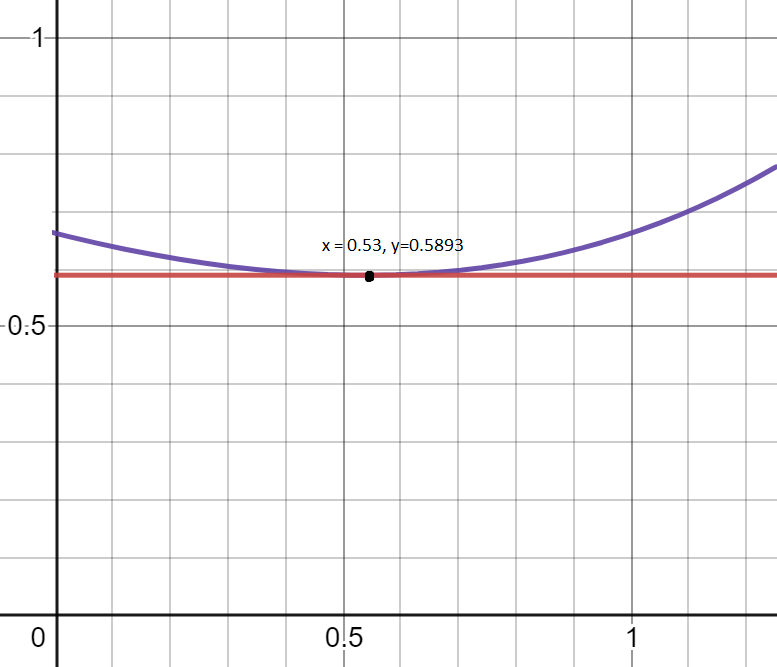}
			\centering
			\caption{With $\beta=0.89$, the straight line (red) is $y=1-e^{-\beta}> 0.5893$ and the curve (purple) is $y=\frac{1}{\beta}\left(e^{\beta(x-1)}-e^{-\beta}\right)+\frac{1-e^{-\beta}}{\beta}(1-x)\geq 0.5893$. Other values of $\beta$ lead to a lower minimum point. }
		\end{figure}
		When $\beta=1$, we have a minimum value of $0.554$.
		\hfill\Halmos\end{proof}

%% file: tex/OCR.tex
\section{Second Approach: Primal-Dual with Online Correlated Rental}
\label{sec:ocr-primal-dual}
In this section, we describe our second approach based on an extension of the online correlated selection (OCS) technique in \cite{fahrbach2020edge}, which we refer to as Online Correlated Rental (\ocs). To see the connection between our problem and the OCR technique (formally defined later), we first describe the architecture of our second proposed algorithm for the \omrr\, problem.  

Our proposed algorithm based is an online primal-dual algorithm maintaining a feasible primal and an infeasible dual assignment for the LP in \Cref{sec:LP} as online vertices arrive. Upon the arrival of an online vertex, we use two algorithms in sequence, referred to as the outer and inner algorithms, to select its offline match. First, the outer algorithm uses the current dual assignment of offline vertices to \emph{suggest} (or equivalently \emph{propose}) either one or two of these offline vertices as potential matches. Note that the outer algorithm ignores the availability of its suggested vertices. When this algorithm suggests a single vertex, which we also refer to as a \emph{deterministic query}, the vertex is matched if it is available. However, when two vertices are suggested, which we also refer to as a \emph{randomized query}, the inner randomized algorithm is then executed to \emph{select} which one to match. This is inspired by a similar idea in \cite{fahrbach2020edge}, where the online integral algorithm pushes its allocation towards more fractional allocations by suggesting two choices rather than one.

One simple candidate for the inner algorithm mentioned above is selecting one of the two vertices independently with probability $\frac{1}{2}$. It can be shown that this approach fails to improve over $0.5$-competitive. To make an improvement, intuitively, the goals are:
\begin{enumerate}[label=(\roman{*})]
    \item Ensuring that a matched resource $i$ is not used in the next $d-1$ time units and hence the final matching decisions are feasible given the reusability of the resources,
    \item For each resource $i$, making sure that its allocation decisions in rounds that can potentially affect each other (i.e., rounds with distance at most $d-1$) are negatively correlated; in other words, if this vertex is suggested and not chosen in round $j$, it will become more likely to be chosen if suggested again in the next $d - 1$ rounds, and vice versa.
\end{enumerate}
The new technical ingredient of our algorithm is the randomized procedure that handles the aforementioned conflicts among the reusable resources suggested by the outer algorithm over time, in a way that it satisfies both properties (i) and (ii) above. This randomized procedure is indeed what we called \ocs\, earlier.  In what follows, we first study the design of the \ocs\, procedure  as an abstract problem in \Cref{sec:OCR}. We then show how to use \ocs~as an inner algorithm with a suitable primal-dual outer algorithm in \Cref{sec:primal-dual}


\subsection{Online correlated rental}
\label{sec:OCR}

In order to design an appropriate OCR procedure, we first formalize the intuitive ``negative correlation criteria'' in property (ii) for the online matching of reusable resources. On the way to formalizing this intuition, we obtain a new interpretation of the OCS procedure for non-reusable resources, which can be extended to the reusable setting. 

\smallskip
\noindent\textbf{\emph {Overview of the technique}} Consider an offline vertex $i$, and suppose its probability of being available at time $j$ is $p_j^i$. If vertex $i$ is suggested by the outer algorithm in a deterministic query, it will increase the expected size of the matching by $p_j^i$. If the vertex $i$ is one of two vertices suggested by the outer algorithm in a randomized query, we can increase the expected size of the matching including $i$ by $p_j^i/2$ if we use a uniform random allocation (with no correlation) to select between the two. Note that the other proposed vertex will also contribute.  The idea behind our OCR procedure is to add a communication across different rounds of this inner selection so that a vertex that is less likely to be selected (although available) at some time is slightly more likely to be available in future \emph{related} rounds and vice versa.
This communication should also happen in an almost symmetric fashion, so that we can maintain a default lower-bound on the increase in the expected size of the matching as if the decisions were made uniformly at random and independently over time. 

The exact implementation of this communication can be thought of as leaving probabilistic temporary messages on offline vertices recording their previous selection decisions. If an offline vertex $i$ is queried at time $j$ and is selected, we will occasionally tag it with this outcome for the next $d -1$ rounds. Should it be suggested again in this window of time, it will be less likely to be selected. Likewise, an offline vertex $i’$ which is suggested but \emph{not} selected will have a chance to be tagged with a message making it \emph{more} likely to be selected in the next $d - 1$ rounds. As long as this successful message-passing happens with a constant probability, the resulting OCR algorithm will satisfy our desired guarantee. 

We start by formalizing the OCR problem and the desired performance guarantee of our OCR procedure, and then later show how to design an algorithm that achieves this guarantee. 

\subsubsection{Basics and definitions}
Let $Q_1, Q_2, \cdots, Q_{|U|}$ be the sequence of suggestions made by the outer algorithm where the outer algorithm suggests $Q_j$  to the inner algorithm at the time $j$, and each $Q_j$ contains one or two offline vertices. Upon arrival of each $Q_j$, the \ocs~algorithm -- which acts as the inner algorithm -- has to choose one offline vertex $i \in Q_j$. The outer algorithm then adds $i$ to the matching if $i$ is available at the time $j$.

 Consider an online algorithm that selects one of the vertices in $Q_j$ uniformly at random. We first analyse the expected size of the matching returned by this algorithm. We say that round $j$ is deterministic if $|Q_j|=1$, and randomized if $|Q_j|=2$. Considering an offline vertex $i$, and let $\Delta \mu^i_j$ be the probability that $i$ is matched at the round $j$. If $i \notin Q_j$, it is clear that the online algorithm will not select $i$, and we have $\Delta \mu^i_j =0$. Now consider a round $j$ where $i \in Q_j$. Since the online algorithm selects one of the vertices in $Q_j$ uniformly at random, the algorithm selects $i$ with the probability of $1/2$ if round $j$ is a randomized round and $1$ if it is a deterministic round. Although the algorithm picks $i$ with the probability of at least $1/2$, it does not necessarily imply that $\Delta \mu^i_j \ge 1/2$. The reason is that vertex $i$ might be matched in the time interval of $[j-d+1,j-1]$, and consequently not be available at the time $j$. We claim that the probability that $i$ is matched in the time interval of $[j-d+1,j-1]$ is $\sum_{t=j-d+1}^{j-1} \Delta \mu^i_t$.

\begin{claim}
    For any time $j$ and any offline vertex $i$, the probability that $i$ is matched in the time interval of $[j-d+1,j-1]$ is $\sum_{t=j-d+1}^{j-1} \Delta \mu^i_t$.
\end{claim}

\begin{proof}{{Proof.}}
Let $X^i_t$ be a random variable which is $1$ if $i$ is matched at the round $t$ and is $0$ otherwise. Therefore, we are looking for the following probability.
\begin{align*}
    \Pr\bigg[ \bigvee_{t=j-d+1}^{j-1} X^i_t\bigg] \,.
\end{align*}
Note that upon matching $i$ at a time $t$, vertex $i$ will not be available before the time $t+d$. Therefore, vertex $i$ can be matched at most once in $[j-d+1,j-1]$, and events $X^i_t$'s are mutually exclusive. We then have,
\begin{align*}
    \Pr\bigg[ \bigvee_{t=j-d+1}^{j-1} X^i_t\bigg] = \sum_{t=j-d+1}^{j-1} \Pr\big[X^i_t\big] \,.
\end{align*}
Note that by the definition we have $\Delta \mu^i_t = \Pr[X^i_t]$, which implies the claim.\hfill\Halmos
\end{proof}

By the claim above the probability that vertex $i$ is available upon arrival of $Q_j$ is $1- \sum_{t=j-d+1}^{j-1} \Delta \mu^i_t$. We define $p^i_j = 1- \sum_{t=j-d+1}^{j-1} \Delta \mu^i_t$ to be the probability that $i$ is available at the time $j$. In the online algorithm that we described, the decision of the algorithm in the round $j$ is independent of its decisions in the previous rounds. Thus, we have
\begin{align*}
    \Delta\mu^i_j = \Pr[ i \text{\ gets selected at the round $j$}] \cdot p^i_j \,.
\end{align*}
Therefore,
\begin{align*}
    \Delta\mu^i_j = 
        \begin{cases}
            0 & i \notin Q_j \\
            p^i_j & i \in Q_j \text{ and } |Q_j|=1 \\
            p^i_j/2 & i \in Q_j \text{ and } |Q_j|=2
        \end{cases}
\end{align*}

As mentioned earlier, the goal in \ocs~ is to outperform this simple randomized algorithm which selects one of the vertices in $Q_j$ uniformly at random, by introducing some negative correlations between the decisions of the randomized rounds. Specifically, \ocs\ guarantees that if vertex $i$ was in $Q_j$ and it was not picked by \ocs, then vertex $i$ should have a slightly better chance to get selected next time that the outer algorithm proposes $i$ to the \ocs. 


Given the sequence $Q_1, Q_2, \cdots, Q_{|U|}$, we use $\prev^i_j$ to denote the previous occurrence of the offline vertex $i$ among the $Q_j$'s before the time step $j$. Specifically, $\prev^i_j$ is the largest index $j' <j$ such that $i \in Q_{j'}$. We also define $\prev^i_j$ to be $-\infty$ if vertex $i$ has not appeared in any of $Q_1, Q_2, \cdots, Q_{j-1}$. We are now ready to give a formal definition of \ocs.   

\begin{definition}[$\gamma$-\ocs]
\label{def:ocr-sym}
An algorithm is $\gamma$-\ocs\ for a $0\le \gamma \le 1$, if for every round $j$ and every offline vertex $i$ the following holds.
\begin{equation}
\label{eq:ocsdef}
    \Delta \mu^i_j \ge 
        \begin{cases}
            0 & i \notin Q_j \\
            p^i_j & i \in Q_j \text{ and } |Q_j|=1 \\
            p^i_j/2 & i \in Q_j \text{ and } |Q_j|=2 \text{ and } j' \le j-d \\
            p^i_j/2 + \gamma (p^i_{j'} - \Delta \mu^i_{j'}) & \textrm{Otherwise},
        \end{cases}
\end{equation}
where in the equation above $j'= \prev^i_j$ and $\Delta \mu^i_j$ denotes the probability that $i$ gets picked as matched at the round $j$. 
\end{definition}

Note that in the definition above the term $p^i_{j'} - \Delta \mu^i_{j'}$ is the probability that $i$ was available at the time $j'$ but it was not picked by \ocs. It is easy to verify that the algorithm that selects one of the vertices in $Q_j$ uniformly at random is $0$-\ocs. Therefore, the main challenge is to show the existence an \ocs\ algorithm for a $\gamma>0$. 

\label{sec:OCR-symmetric}
We now show the existence of an $1/32$-\ocs\  algorithm. 
\begin{proposition}
\label{prop:OCR}
Algorithm~\ref{alg:ocs-sym} is an $1/32$-\ocs\ online algorithm (as in \Cref{def:ocr-sym})
\end{proposition}
\SetInd{0.5em}{0.5em}
\begin{algorithm}[htb]
\SetAlgoNoLine
        Upon receiving a proposal $Q_j$\;
        \eIf{$|Q_j|=1$}
        {
            Let assume $Q_j=\{i_1\}$\;
            $\tau_{i_{1},[j+1, j+d-1]} = \un$\;
            \textbf{Select} $i_1$\;
        }
        {
        Let assume $Q_j = \{i_1,i_2 \}$ \;
        With the probability of $1/2$, round $j$ is a sender:\\ \Indp
        Draw $\ell, m \in \{1,2\}$ uniformly at random \;
        $\tau_{i_{\bar m},[j+1, j+d-1]} = \un$ \;
        \eIf{$m=\ell$}
        {
        $\tau_{i_{m},[j+1, j+d-1]} = \sel$\;
        }
        {
        $\tau_{i_{m},[j+1, j+d-1]}= \nsel$ \;
        } \Indm
        With the probability of $1/2$, round $j$ is a receiver: \\ \Indp
        Draw $m \in \{1,2\}$ uniformly at random \;
        \eIf{$\tau_{i_m,j}=\sel$}
        {
            $\ell= \bar m$\;
        }
        {
            Draw $\ell \in \{1,2\}$ uniformly at random\;
        }
        $\tau_{i_1,[j+1, j+d-1]}= \un$\;
        $\tau_{i_2,[j+1, j+d-1]} = \un$\; \Indm
        \textbf{Select} $i_\ell$\;
        }
		\caption{$1/32$-\ocs\ Algorithm}
		\label{alg:ocs-sym}
	\end{algorithm}

\subsubsection{OCR algorithm and analysis}

Before proceeding to the proof of \Cref{prop:OCR}, we provide an overview of \Cref{alg:ocs-sym}. The algorithm maintains a state variable $\tau_{i,j}$ for every offline vertex $i$ and time $j$. Let $j' = \prev^i_j$, if $\tau_{i,j}$ is either \textit{selected} or \textit{not-selected} it means that the result of $Q_{j'}$ can be used at the time $j$ to give us a negative correlation. Otherwise, if $\tau_{i,j}$ is \textit{unknown}, it means that we can not use the result of the last occurrence of $i$. In Algorithm \ref{alg:ocs-sym}, with a slight abuse of notation we use $\tau_{i,[l,r]}$ to denote all state variables $\tau_{i,l},\tau_{i,l+1}, \cdots, \tau_{i,r}$.

Suppose that the algorithm receives a pair $Q_j=\{i_1, i_2\}$ at the time $j$. Then the algorithm decides to be a \textit{receiver} or \textit{sender} at this round uniformly at random with the probability of $1/2$. In a \textit{sender} round, the algorithm uses a fresh random bit to select $i_\ell$ where $\ell \in \{1,2\}$ is chosen uniformly at random. the algorithm also uses a new random bit to select $m \in \{1,2\}$ at random. The algorithm then updates state variables $\tau_{i_m,[j+1,j+d-1]}$ to reflect whether we have selected $i_{m}$ in this round. Specifically, we set state variables $\tau_{i_m,[j+1,j+d-1]}$ to be \textit{selected} if $\ell=m$ and we set $\tau_{i_m,[j+1,j+d-1]}$ to be \textit{not-selected} otherwise. We also reset the state variables for $i_{\bar m}$ by setting $\tau_{i_{\bar{m}},[j+1,j+d-1]}$ to be \textit{unknown} where $\bar m$ is an abbreviation for $3-m$. We remark that in a sender round, we only change the state variables for the next $d-1$ rounds. Therefore, the decision of the algorithm in a sender round can only affect the next $d-1$ rounds.

In a receiver round, the \ocs\ uses the state variables that have been set in previous rounds. The algorithm chooses $m \in \{1,2\}$ uniformly at random and uses the state variables of vertex $i_m$ to provide a negative correlation. Consider the case $\tau_{i_m,j}$ is \textit{not-selected}. This means that the \ocs\ did not pick $i_m$ the previous time this vertex was proposed to the algorithm. In this case, the algorithm selects $i_m$ for the round $j$. When $\tau_{i_m,j}$ is \textit{selected}, the case is similar and the algorithm selects $i_{\bar m}$ at this round. When $\tau_{i_m,j}$ is \textit{unknown}, the algorithm selects $\ell \in \{0,1\}$ uniformly at random, and selects $i_\ell$ at this round. At the end of a receiver round the algorithm resets the state variables for the future rounds by setting $\tau_{i_{1},[j+1,j+d-1]}$ and $\tau_{i_{2},[j+1,j+d-1]}$ to be \textit{unknown}.

What we have discussed so far was our \ocs\ algorithm for randomized rounds. Consider a deterministic round and suppose that the algorithm receives $Q_j= \{i_1 \}$. In this case the algorithm selects $i_1$, and resets the state variables of $i_1$ for the future rounds by setting $\tau_{i_{1},[j+1,j+d-1]}$ to be \textit{unknown}.

\paragraph{Proof of \Cref{prop:OCR}} We now show that Algorithm \ref{alg:ocs-sym} is $1/32$-\ocs. Our analysis relies on a simple fact. Let $\mu^i(Q_1, Q_2, \cdots, Q_{|U|})$ be expected the number of times that the algorithm matches vertex $i$ when it receives sequence of vertices $Q_1, Q_2, \cdots, Q_{|U|}$. We show that we always have $\mu^i(Q_1, Q_2, \cdots, Q_{|U|}) = \mu^i(Q_{|U|}, Q_{|U|-1}, \cdots, Q_1)$. This fact shows that if Algorithm \ref{alg:ocs-sym} receives the sets $Q_1, Q_2, \cdots, Q_{|U|}$ in the reversed order, it matches every offline vertex the same number times in expectation. For the simplicity of the presentation, we use $\vec{Q}_{[l,r]}$ as an abbreviation for the sequence $Q_l, Q_{l+1}, \cdots, Q_r$ when $l \le r$, and for the sequence $Q_l, Q_{l-1}, \cdots, Q_r$ when $l > r$.

\begin{claim}
\label{clm:reverse}
For any sets of proposed vertices $Q_1, Q_2, \cdots, Q_{\ell}$, and every offline vertex $i$ the following holds.
$$\mu^i(Q_{[1,\ell]}) = \mu^i(Q_{[\ell,1]}) \,.$$
\end{claim}
The proof of the claim above is deferred to Appendix~\ref{apx:claim-rev}.

We now show that Algorithm \ref{alg:ocs-sym} is $1/32$-\ocs. Let $Q_1, Q_2, \cdots, Q_{|U|}$ be the sets proposed by the outer algorithm. We show that the expected matching returned by Algorithm \ref{alg:ocs-sym} satisfies \cref{eq:ocsdef}. By the definition of $\mu^i$, for any round $j$, $\mu^i(Q_1, \cdots, Q_j)$ is the expected number of times that the algorithm matches $i$ for the first $j$ rounds. Therefore, $\mu^i(Q_1, \cdots, Q_j) - \mu^i(Q_1, \cdots, Q_{j-1})$ is the expected increment in the size of the matching in round $j$, which is the probability that the algorithm matches $i$ at the round $j$. Therefore, we have

\begin{align}
    \label{eq:deltamu}
    \Delta \mu^i_j= \mu^i(Q_1, \cdots, Q_j) - \mu^i(Q_1, \cdots, Q_{j-1}) \,.
\end{align}
By applying Claim \ref{clm:reverse}, we get the following.
\begin{align}
    \label{eq:deltamurev}
    \Delta \mu^i_j= \mu^i(Q_j, \cdots, Q_1) - \mu^i(Q_{j-1}, \cdots, Q_1) \,.
\end{align}
We can re-write (\ref{eq:deltamurev}) as 
\begin{align}
\label{eq:deltamusimp}
\Delta \mu^i_j = \mu^i(\vec{Q}_{[j,1]}) - \mu^i(\vec{Q}_{[j-1,1]})
\end{align}

Consider Algorithm \ref{alg:ocs-sym} for the input sequence $\vec{Q}_{[j,1]}$ and an offline vertex $i$,  we derive a lower bound on the expected number of times that the algorithm matches $i$.  First, we claim that in order to find the expected size of the matching of vertex $i$, we can only look at the decisions of the algorithm for rounds that $i$ is proposed. If we consider a randomized round $Q_j=\{i_1, i_2\}$, the algorithm picks each of $i_1$ and $i_2$ with the probability of $1/2$ when the probability is taken over all random decisions of the algorithm. Algorithm \ref{alg:ocs-sym} just introduces a slight negative correlation between its decision for different rounds. Although the decision of the algorithm for a round $Q_j$ which contains vertex $i$ might be correlated to other rounds that does not contain $i$, these correlations does not affect the number of times that the algorithm matches $i$. Therefore, for the analysis of the algorithm, we can only consider the rounds that contain $i$ and the correlations between these rounds. 

We now consider the different cases of definition of \ref{eq:ocsdef}, and we show that our algorithm satisfies the properties of \ocs. 

\begin{itemize} [leftmargin=*]
    \item \textbf{Case 1; $i \notin Q_j$:} As we discussed earlier we can only consider rounds that $i$ is proposed to the algorithm. Since $i$ is not proposed in round $j$, the algorithm does not match $i$, and goes to the next round. Therefore, in this case we have $\mu^i(\vec{Q}_{[j,1]}) = \mu^i(\vec{Q}_{[j-1,1]})$. By (\ref{eq:deltamusimp}), we have
    \begin{align*}
        \Delta \mu^i_j = \mu^i(\vec{Q}_{[j,1]}) - \mu^i(\vec{Q}_{[j-1,1]}) = 0 \,.
    \end{align*}
    \item \textbf{Case 2; $i \in Q_j$ and $|Q_j|=1$:} In this case, the algorithm is given a deterministic round. Therefore, the algorithm picks and matches $i$ with the probability of $1$. After matching $i$, this vertex becomes available for matching again only after $d$ time units when set $Q_{j-d}$ arrives. In this case the algorithm cannot match $i$ when the sets $Q_{[j-1, j-d+1]}$ arrive. Therefore the decisions and correlations of the algorithm for these rounds does not affect the expected size of the matching. Therefore, we have $\mu^i(\vec{Q}_{[j,1]}) = 1+ \mu^i(\vec{Q}_{[j-d,1]})$. We then have,
    \begin{align}
    \label{eq:ocssymdet}
        \Delta \mu^i_j &= \mu^i(\vec{Q}_{[j,1]}) - \mu^i(\vec{Q}_{[j-1,1]}) = 1+ \mu^i(\vec{Q}_{[j-d,1]}) - \mu^i(\vec{Q}_{[j-1,1]})
    \end{align}
    In the following observation we show that $\mu^i(\vec{Q}_{[j-1,1]}) - \mu^i(\vec{Q}_{[j-d,1]}) = \sum_{t=j-d+1}^{j-1} \Delta \mu^i_t$.
    
    \begin{observation}
        \label{obs:telsum}
        For any $j'<j$, 
        \begin{align*}
            \mu^i(\vec{Q}_{[j,1]}) - \mu^i(\vec{Q}_{[j',1]}) = \sum_{t=j'+1}^{j} \Delta \mu^i_t.
        \end{align*}
    \end{observation}
    \begin{proof}{Proof.}
        We can re-write $\mu^i(\vec{Q}_{[j,1]}) - \mu^i(\vec{Q}_{[j',1]})$ as follows.
        \begin{align*}
            \mu^i(\vec{Q}_{[j,1]}) - \mu^i(\vec{Q}_{[j',1]}) = \sum_{t=j'+1}^{j} \big(\mu^i(\vec{Q}_{[t,1]}) - \mu^i(\vec{Q}_{[t-1,1]})\big) \,.
        \end{align*}
        Using equation (\ref{eq:deltamusimp}), we then have
        \begin{align*}
            \mu^i(\vec{Q}_{[j,1]}) - \mu^i(\vec{Q}_{[j',1]}) &= \sum_{t=j'+1}^{j} \big(\mu^i(\vec{Q}_{[t,1]}) - \mu^i(\vec{Q}_{[t-1,1]})\big) =\sum_{t=j'+1}^{j} \Delta \mu^i_t\,.
        \end{align*}\hfill\Halmos
    \end{proof}
    
    By the observation above and (\ref{eq:ocssymdet}) we have
    \begin{align*}
        \Delta \mu^i_j & = 1+ \mu^i(\vec{Q}_{[j-d,1]}) - \mu^i(\vec{Q}_{[j-1,1]}) = 1- \sum_{t=j-d+1}^{j-1} \Delta \mu^i_t \,.
    \end{align*}
    Note that by the definition $1- \sum_{t=j-d+1}^{j-1} \Delta \mu^i_t$ is equal to $p^i_j$. Therefore, $\Delta \mu^i_j = p^i_j$ which completes the proof for this case.
    \item \textbf{Case 3; $i \in Q_j$,  $|Q_j|=2$, and $\prev^i_j \le j-d$:} Let $j' = \prev^i_j$, then $Q_{j'}$ is the next time that vertex $i$ is proposed to the algorithm in the sequence $\vec{Q}_{[j-1,1]}$. Recall that Algorithm \ref{alg:ocs-sym} upon making a decision for a sender round, changes the state variables for the next $d-1$ rounds. In the case that $j' \le j-d$, the algorithm does not change the state variables of $Q_{j'}$. Consider the sequence of proposals $Q_{[j,1]}$ and consider the algorithm when it receives $Q_j$. The algorithm selects $i$ with the probability of $1/2$ regardless it is a \textit{sender} or a \textit{receiver} round (Since $Q_j$ is the first round in the backward sequence). Also the sender round does not change the state variables for the next time vertex $i$ gets proposed to the algorithm. Therefore with the  probability of $1/2$, algorithm selects $i$, and we get the expected matching of size $1+ \mu^i(\vec{Q}_{[j-d,1]})$. Furthermore, the algorithm does not select $i$ with the probability of $1/2$. In this case the algorithm basically goes to next round and we get the expected matching of size $ \mu^i(\vec{Q}_{[j-1,1]})$. Therefore,
    \begin{align}
        \label{eq:ocssymc3}
        \mu^i(\vec{Q}_{[j,1]}) = \frac{1}{2} \cdot (1+ \mu^i(\vec{Q}_{[j-d,1]}))+ \frac{1}{2} \cdot \mu^i(\vec{Q}_{[j-1,1]}) \,.
    \end{align}
    Hence,
    \begin{align*}
        \Delta \mu^i_j &= \mu^i(\vec{Q}_{[j,1]}) - \mu^i(\vec{Q}_{[j-1,1]}) & \text{By (\ref{eq:deltamusimp}).} \\
        &= \frac{1}{2} \cdot \big(1+ \mu^i(\vec{Q}_{[j-d,1]}) - \mu^i(\vec{Q}_{[j-1,1]})\big) & \text{By (\ref{eq:ocssymc3}).} \\
        &= \frac{1}{2} \cdot \big(1- \sum_{t=j-d+1}^{j-1} \Delta \mu^i_t\big) = \frac{p^i_j}{2}  & \text{Observation \ref{obs:telsum}.} \,.
    \end{align*}
    \item \textbf{Case 4; $i \in Q_j$, $|Q_j|=2$, and $\prev^i_j > j-d$:} This is the only case that algorithm uses the negative correlation between the decisions for different rounds to get a better solution. Let $j' = \prev^i_j$, We further divide this case into two subcases.
    \begin{itemize} [leftmargin=12px]
        \item \textbf{$Q_{j'}$ is a deterministic round:}
        First note that when $Q_{j'}$ is deterministic round, we have $\Delta \mu^i_{j'} = p^i_{j'}$ according to what we have already discussed. Also, the decision of a deterministic round does not rely on the state variables, and Algorithm \ref{alg:ocs-sym} always resets state variables after a deterministic round. Considering the backward sequence $Q_{[j,1]}$, the algorithm selects $i$ with the probability of $1/2$ when it receives $Q_j$ regardless of if it is a \textit{sender} or a \textit{receiver} round. Therefore with the  probability of $1/2$, our algorithm selects $i$, and we get the expected matching of size $1+ \mu^i(\vec{Q}_{[j-d,1]})$. Also, with the  probability of $1/2$, our algorithm does not select $i$ when it receives $Q_j$. In this case vertex $i$ only has a higher chance to get selected next times this vertex is proposed to the algorithm. This depends on the realized edges incident to $Q_j$. Therefore, if the algorithm does not select $i$, the expected size of the matching in the remaining sequence is at least $\mu^i(\vec{Q}_{[j-1,1]})$. Therefore,
        \begin{align*}
        \mu^i(\vec{Q}_{[j,1]}) \ge \frac{1}{2} \cdot (1+ \mu^i(\vec{Q}_{[j-d,1]}))+ \frac{1}{2} \cdot \mu^i(\vec{Q}_{[j-1,1]})  \,.
    \end{align*}
    We then have
        \begin{align*}
        \Delta \mu^i_j &= \mu^i(\vec{Q}_{[j,1]}) - \mu^i(\vec{Q}_{[j-1,1]}) & \text{By (\ref{eq:deltamusimp}).} \\
        &\ge \frac{1}{2} \cdot \big(1+ \mu^i(\vec{Q}_{[j-d,1]}) - \mu^i(\vec{Q}_{[j-1,1]})\big)\\
        &= \frac{1}{2} \cdot \big(1- \sum_{t=j-d+1}^{j-1} \Delta \mu^i_t\big) & \text{Observation \ref{obs:telsum}.} \\
        &= \frac{p^i_j}{2} = \frac{p^i_j}{2} + \gamma (p^i_{j'} - \Delta \mu^i_{j'})\,,
    \end{align*}
        where the last equality relies on the fact that  $ p^i_{j'} = \Delta \mu^i_{j'}$ and  $ p^i_{j'} -  \Delta \mu^i_{j'} = 0$.
        \item \textbf{$Q_{j'}$ is a randomized round:}   Consider the backward sequence $Q_{[j,1]}$, and let $\mu^i_{\receiver}(\vec{Q})$ be the expected size of the matching for vertex $i$, given that the round $Q_j$ is a receiver round, and let $\mu^i_{\sender}(\vec{Q})$ be the expected size of the matching given that the $Q_j$ is a sender round. Note that at each randomized round the algorithm chooses to be a sender or a receiver uniformly at random. Therefore for any input sequence $\vec{Q}$ such that the first round of $\vec{Q}$ is a randomized round, we have
            \begin{align}
            \label{eq:mu-rs}
            \mu^i(\vec{Q}) = \frac{1}{2} \cdot \mu^i_{\receiver}(\vec{Q})+ \frac{1}{2} \cdot \mu^i_{\sender}(\vec{Q}) \,.
        \end{align}
        
        Consider the input sequence $\vec{Q}_{[j,1]}$ and suppose that $Q_j$ is a receiver round. Since this is a first round in the backward sequence and all state variables are default values, the algorithm just selects $i$ with the probability of $1/2$ and proceeds to the next round. Similar to what we discussed for the previous case, if the algorithm selects $i$, the expected size of the matching is $1+ \mu^i(\vec{Q}_{[j-d,1]})$. Otherwise, it is $ \mu^i(\vec{Q}_{[j-1,1]})$. Thus,
        \begin{align}
        \label{eq:mureciver}
            \mu^i_{\receiver}(\vec{Q}_{[j,1]}) = \frac{1}{2} \cdot (1+ \mu^i(\vec{Q}_{[j-d,1]}))+ \frac{1}{2} \cdot \mu^i(\vec{Q}_{[j-1,1]}) \,.
        \end{align}
        
        Now consider a sender round $Q_j$. The algorithm selects $i$ with the probability of $1/2$. In this case the algorithm cannot match $i$ again within the next $d-1$ rounds. Therefore, in this case the expected size of the matching is $1+ \mu^i(\vec{Q}_{[j-d,1]})$. Otherwise suppose that the algorithm did not select $i$ which happens with the probability of $1/2$. Let $Q_j= \{i_1, i_2\}$. In this case algorithm selects $m \in \{0,1\}$ uniformly at random and changes the state variables for $i_m$ to create a negative correlation for the next rounds. Suppose $i_m \neq i$ which happens with the probability of $1/2$. In this case, the algorithm might have a realized edge between this round and one of the other rounds that $i$ is proposed to the algorithm. In this case the algorithm has a better chance to select $i$ for the next rounds. However, the existence of this realized edge cannot be guaranteed and we can say that in this case the algorithm returns a matching with the expected size of at least $\mu^i(\vec{Q}_{[j-1,1]})$ in the remaining sequence. In the remaining of the analysis we assume that the algorithm does not pick $i$ for round $Q_j$, and $i_m = i$ for this round. 
        
        In this case the algorithm sets the state variables $\tau_{i}$  to \textit{not-selected} for the next $d-1$ rounds to introduce a negative correlation. Now consider the round that the algorithm receives $i$ which is$Q_{j'}$. If round $Q_{j'}$ is sender round, the algorithm does not use the previous state variables and we get an expected matching size of $\mu^i_{\sender}(\vec{Q}_{[j',1]})$. Now Consider a receiver round $Q_{j'}$. Let $Q_{j'}= \{i'_1, i'_2\}$. In this case algorithm selects $m \in \{0,1\}$ uniformly at random and uses the state variables for $i'_m$ to ensure a negative correlation.  If $i'_m \neq i$, the algorithm might have a realized edge between rounds $Q_j$ and $Q_{j'}$ which causes a negative correlation. However, the existence of such an edge cannot be guaranteed. Therefore, in this case we can assume that algorithm does not uses the state variables. Thus, we get at a matching of size least $\mu^i_{\receiver}(\vec{Q}_{[j',1]})$. However, if $i'_m$ is equal to $i$, the algorithm ensures a negative correlation between rounds $Q_j$ and $Q_{j'}$. Therefore the algorithm selects $i$, and we get an expected matching size of $1+ \mu^i(\vec{Q}_{[j'-d,1]})$.
        
        Putting all together, for a sender round $Q_j$ we have
        
        \begin{align}
        \label{eq:musender}
             \mu^i_{\sender}(\vec{Q}_{[j,1]})& \ge \frac{1}{2} \cdot \big(1+ \mu^i(\vec{Q}_{[j-d,1]})\big) + \frac{1}{4} \cdot \mu^i(\vec{Q}_{[j-1,1]}) \nonumber \\
             &+  \frac{1}{8} \cdot \mu^i_{\sender}(\vec{Q}_{[j',1]}) + \frac{1}{16} \cdot \mu^i_{\receiver}(\vec{Q}_{[j',1]}) + \frac{1}{16} \cdot \big(1+ \mu^i(\vec{Q}_{[j'-d,1]})\big)
        \end{align}

        By Equation (\ref{eq:mu-rs}), we then have
        
        \begin{align*}
            \mu^i(\vec{Q}_{[j,1]})  &=  \frac{1}{2} \cdot \mu^i_{\receiver}(\vec{Q}_{[j,1]})+ \frac{1}{2} \cdot \mu^i_{\sender}(\vec{Q}_{[j,1]}) \\
            & \ge \frac{1}{4} \cdot (1+ \mu^i(\vec{Q}_{[j-d,1]}))+ \frac{1}{4} \cdot \mu^i(\vec{Q}_{[j-1,1]}) \\
            & + \frac{1}{4} \cdot \big(1+ \mu^i(\vec{Q}_{[j-d,1]})\big) + \frac{1}{8} \cdot \mu^i(\vec{Q}_{[j-1,1]}) \\
            &+  \frac{1}{16} \cdot \mu^i_{\sender}(\vec{Q}_{[j',1]}) + \frac{1}{32} \cdot \mu^i_{\receiver}(\vec{Q}_{[j',1]}) \\
            & + \frac{1}{32} \cdot \big(1+ \mu^i(\vec{Q}_{[j'-d,1]})\big) & \text{By (\ref{eq:mureciver}) and (\ref{eq:musender}).}
        \end{align*}
        By simplifying the inequality above we get
        \begin{align*}
            \mu^i(\vec{Q}_{[j,1]}) & \ge \frac{1}{2} \cdot (1+ \mu^i(\vec{Q}_{[j-d,1]}))+ \frac{3}{8} \cdot \mu^i(\vec{Q}_{[j-1,1]}) \\
            &+  \frac{1}{16} \cdot \mu^i_{\sender}(\vec{Q}_{[j',1]}) + \frac{1}{32} \cdot \mu^i_{\receiver}(\vec{Q}_{[j',1]}) \\
            & + \frac{1}{32} \cdot \big(1+ \mu^i(\vec{Q}_{[j'-d,1]})\big)
        \end{align*}
        
        Note that by Equation (\ref{eq:mu-rs}), we have $\mu^i(\vec{Q}_{[j',1]}) = \frac{1}{2} \cdot \mu^i_{\sender}(\vec{Q}_{[j',1]}) + \frac{1}{2} \cdot \mu^i_{\receiver} (\vec{Q}_{[j',1]})$. Thus, we can re-write the inequality above as
        \begin{align*}
            \mu^i(\vec{Q}_{[j,1]}) & \ge \frac{1}{2} \cdot \big(1+ \mu^i(\vec{Q}_{[j-d,1]}) + \mu^i(\vec{Q}_{[j-1,1]})\big)  + \frac{1}{32} \cdot \big(1+ \mu^i(\vec{Q}_{[j'-d,1]}) - \mu^i_{\receiver}(\vec{Q}_{[j',1]})\big) \,.
        \end{align*}
        
        Therefore,
        \begin{align}
        \label{eq:ocsc4delta}
            \Delta \mu^i_j &= \mu^i(\vec{Q}_{[j,1]}) - \mu^i(\vec{Q}_{[j-1,1]}) & \text{By (\ref{eq:deltamusimp}).} \nonumber \\
            &\ge \frac{1}{2} \cdot \big(1+ \mu^i(\vec{Q}_{[j-d,1]}) - \mu^i(\vec{Q}_{[j-1,1]})\big) \nonumber \\
            & + \frac{1}{32} \cdot \big(1+ \mu^i(\vec{Q}_{[j'-d,1]}) - \mu^i_{\receiver}(\vec{Q}_{[j',1]})\big)
        \end{align}
    \end{itemize}
    
    We complete the proof using the claim below. The proof is available in Appendix \ref{appx:missing}.
    
    \begin{claim} \label{claim:receiversender}
        For any backward sequence $\vec{Q}_{[j,1]}$ where $Q_j$ is a randomized round, the following holds.
        \begin{align*}
            \mu^i_{\receiver}(\vec{Q}_{[j,1]}) \le \mu^i(\vec{Q}_{[j,1]}) \le \mu^i_{\sender}(\vec{Q}_{[j,1]}) \,. 
        \end{align*}
    \end{claim}
    
    It follows from the claim above and Inequality (\ref{eq:ocsc4delta}) that
    
    \begin{align*}
        \Delta \mu^i_j &\ge \frac{1}{2} \cdot \big(1+ \mu^i(\vec{Q}_{[j-d,1]}) - \mu^i(\vec{Q}_{[j-1,1]})\big) \\
            & + \frac{1}{32} \cdot \big(1+ \mu^i(\vec{Q}_{[j'-d,1]}) - \mu^i_{\receiver}(\vec{Q}_{[j',1]})\big) \\
            &\ge \frac{1}{2} \cdot \big(1+ \mu^i(\vec{Q}_{[j-d,1]}) - \mu^i(\vec{Q}_{[j-1,1]})\big) \\
            & + \frac{1}{32} \cdot \big(1+ \mu^i(\vec{Q}_{[j'-d,1]}) - \mu^i(\vec{Q}_{[j',1]})\big) \\
            &= \frac{1}{2} \cdot \big(1 - \sum_{t=j-d+1}^{j-1} \Delta \mu^i_t\big) \\
            & + \frac{1}{32} \cdot \big(1 - \sum_{t=j'-d+1}^{j'} \Delta \mu^i_t \big) & \text{Observation \ref{obs:telsum}.} \\
            &= \frac{1}{2} \cdot p^i_j + \frac{1}{32} \cdot (p^i_{j'} - \Delta \mu^i_{j'}) \,,
    \end{align*}
    which completes the analysis and shows that Algorithm \ref{alg:ocs-sym} is $1/32$-\ocs.
\end{itemize}

%% file: tex/primal-dual-OCR.tex
\subsection{Primal-dual algorithm}
\label{sec:primal-dual}
We now present our final primal-dual algorithm based on the OCR algorithm proposed in \Cref{sec:OCR}. We then provide a primal-dual analysis of the algorithm to show the competitive ratio of $\Gamma\approx 0.505$.

\smallskip
\noindent\textbf{\emph{Overview}} Given the type of guarantees by the OCR procedure as in \Cref{def:ocr-sym}, we develop a primal-dual algorithm and its analysis. Our method follows a similar logic as the recent primal-dual techniques in \cite{huang2020adwords} and \cite{fahrbach2020edge}. We end up with identifying constraints that should be satisfied by our competitive ratio $\Gamma$ and our dual assignments in order to guarantee approximate dual feasibility. We finish the analysis by setting up a factor revealing linear program that maximizes $\Gamma$ given these constraints, and find the optimal solution (up to any desired precision) by appropriate discretization and running a computer program. 

\subsubsection{Description of the algorithm}
\label{sec:primal-dual-alg}
Our OCR-based Primal-Dual (OPD) algorithm is represented in Algorithm~\ref{alg:primal-dual}. At each round, the algorithm selects a set of at most two vertices, and propose these vertices to the \ocs\ algorithm. Consider the round $j$ of the algorithm and let $Q_1, Q_2, \ldots, Q_{j-1}$ be the previous set of vertices proposed by the algorithm, and let $Q_j$ be the new set proposed by the algorithm. We define the marginal contribution of the $Q_j$ as follows.

Algorithm \ref{alg:primal-dual} and its analysis relies on the primal dual LP of \omrr. Our algorithm maintains an online primal dual solution of \omrr. In the primal solution we set $x_{i,j} = \Delta \mu^i_j$. Therefore, the size of the primal solution is always equal to the expected size of the matching returned by the algorithm. Let $Q_j$ be the query proposed by the algorithm at the time $j$. Let $i \in Q_j$, be a vertex proposed to \ocs\ at the step $j$ of the algorithm. Proposing this vertex increases the expected size of dual solution by $\Delta \mu^i_j$. Our algorithm maintains a dual solution by increasing dual variables $
\alpha_{i,j}$ and $\beta_j$ by $\Delta \alpha^i_{j}$ and $\Delta \beta^i_j$ respectively. Our algorithm always guarantee that $\Delta \alpha^i_{j} +\Delta \beta^i_j = r_i \cdot \Delta \mu^i_j$. Thus, the size of the dual solution is always equal to the size of primal solution.  Later in this section, we explain how we determine values of $\Delta \alpha^{i}_j$ and $\Delta \beta^i_j$.

At each round, our online primal dual algorithm finds a set of at most two vertices that maximizes $\beta_j$, and proposes this set to \ocs. Note that our primal-dual algorithm does not look at the random decisions of \ocs. Nevertheless, we can assume that we find an online matching which is constructed in the following way. Whenever \ocs\ selects a vertex $i$, we add that vertex to the matching if vertex $i$ is available. This vertex will not be available for the next $d$ time-units, although \ocs\ might pick this vertex again during this time period.

Consider a randomized round $Q_j$, and an offline vertex $i \in Q_j$. Let $j'= \prev^i_j$, thus $Q_{j'}$ is the previous occurrence of $i$. Suppose that $j' > j-d$. Then, a $\gamma$-\ocs\ algorithm guarantees that $i$ is selected with at least probability of  $p^i_j/2 + \gamma (p^i_{j'} - \Delta \mu^i_{j'})$. Recall that $p^i_{j'} - \Delta \mu^i_{j'}$ is the probability that $i$ was available at time $j'$ but is not picked by \ocs.  Let $u^i_j= p^i_{j'} - \Delta \mu^i_{j'}$, and let ${q}^i_j$ be the probability that $i$ was not available at the time $j'$ but became available due to re-usability before or at the time $j$. It is then clear that the probability that $i$ is available at the time $j$ is equal to $u^i_j+ q^i_j$. Therefore,
\begin{align*}
    p^i_j=u^i_j + q^i_j \,.
\end{align*}
We can then rewrite the bound guaranteed by \ocs\ as follows.
\begin{align*}
   p^i_j/2 + \gamma (p^i_{j'} - \Delta \mu^i_{j'}) &= p^i_j/2 + \gamma u^i_j
   = (u^i_j + q^i_j)/2 + \gamma u^i_j
   = (1/2+\gamma) u^i_j + q^i_j/2 \,.
\end{align*}

This shows that whenever vertex $i$ was available  at the time $j'$ but it is not matched by the algorithm, the \ocs\ algorithm guarantees that it will be picked next time with a probability slightly larger than $1/2$. However, if it is the first time that $i$ is proposed after it became available again, the algorithm still picks $i$ with the probability of $1/2$. We set the dual variables based on the values of $u^i_j$ and $q^i_j$. Intuitively, we assign a different weight to dual variables whenever $i$ is proposed to \ocs\ for the first time after becoming available again. For a vertex $i$ in $Q_j$, we set $\Delta \beta^i_j = \beta_1 \cdot r_i \cdot q^i_j/2 + \beta_2 \cdot r_i \cdot (\Delta \mu^i_j - q^i_j/2)$, where $\beta_1 \ge \beta_2$ are constants that we optimize through this section. We also set $\Delta \alpha^i_j = (1-\beta_1) \cdot r_i \cdot q^i_j/2 + (1-\beta_2) \cdot r_i \cdot (\Delta \mu^i_j - q^i_j/2)$. Note that for the consistency of our notations we set $q^i_j= p^i_j$ whenever $\prev^i_j \le j-d$ (In this case we also have $p^i_j=1$).

\SetInd{0.5em}{0.5em}
\begin{algorithm}[H]
		\SetAlgoNoLine
		\For{each online vertex $j$}
		{
		    Find a set $Q_j \subseteq N(j)$ with the size at most two vertices that maximizes $\beta_j$, where $N(j)$ is the set of neighbors of the online vertex $j$\;
		    Propose $Q_j$ to a $\gamma$-\ocs\ algorithm, and update primal-dual variables accordingly \;
		    Let $i$ be the vertex picked by the \ocs. Add $(i,j)$ to the matching if $i$ is available at the time $j$ \;
		}
		\textbf{Return} the matching \;
		\caption{OCR-based Primal-Dual (OPD)}
		\label{alg:primal-dual}
\end{algorithm}
\subsubsection{Primal-Dual Analysis}
\label{sec:primal-dual-analysis}
It is clear from our dual assignments in \Cref{sec:primal-dual-alg} that we always have $\Delta \beta^i_j + \Delta \alpha^i_j = r_i \cdot \Delta \mu^i_j$. Therefore, the size of the primal and dual solutions are always equal. In this section we show that for the proper choices of $\beta_1$ and $\beta_2$, our algorithm satisfies the approximate dual feasibility as it is defined in Lemma \ref{lem:primal-dual}.  We begin our analysis by simple a observation about the dual variables.

\begin{claim}
\label{obs:lb-dual}
Let $i$ be a vertex in $Q_j$, and assume that we have assigned the dual variables according to the rule we discussed earlier. We then have
\begin{align*}
    \sum_{t=\max\{j-d+1,1\}}^{j-1} \alpha_{i,t} \ge (1-\beta_1) \cdot r_i \cdot \varphi^i_j + (1-\beta_2) \cdot r_i \cdot (1-p^i_j - \varphi^i_j) \,,
\end{align*}
where $\varphi^i_j = \min\{1-p^i_j, (1-q^i_j)/2\}$.
\end{claim}
\begin{proof}{Proof.}
Consider a vertex $i$, if we consider time steps $j-d+1, j-d+2, \cdots, j-1$, this vertex has a contribution of $1-p^i_j$ to the size of the matching. Let $T \subseteq [j-d+1,j-1]$ be all the time steps between $j-d+1$ and $j-1$ (inclusive) such that the algorithm has proposed $i$ to \ocs. Now consider a time $t \in T$ where $i$ is proposed to \ocs. If it is the first time that $i$ gets proposed to \ocs\ after becoming available (which happens with the probability of $q^i_t$), we increase $\Delta \alpha^i_t$ by $(1-\beta_1)\cdot r_i \cdot q^i_t/2$. If vertex $i$ gets matched with a probability greater than $q^i_t/2$ at the time $t$ or if it was available at the time $t$ and gets matched in later steps, we add $(1-\beta_2)$ portion of its marginal contribution to $\Delta \alpha^i_t$. Note that $(1-\beta_2)$ is larger than $(1-\beta_1)$ since $\beta_1 \ge \beta_2$. Thus,
\begin{align*}
\sum_{t=\max\{j-d+1,1\}}^{j-1} \alpha_{i,t} = \sum_{t \in T} \alpha_{i,t} \ge (1-\beta_1) \cdot r_i \cdot \sum_{t \in T} q^i_t/2 + (1-\beta_2) \cdot r_i \cdot (1- p^i_j -  \sum_{t \in T} q^i_t/2)
\end{align*}
We now claim that $\sum_{t \in T} q^i_t/2 \le \varphi^i_j$ where $\varphi^i_j=  \min\{1-p^i_j, (1-q^i_j)/2\}$ which proves the claim. It is easy to see that $\sum_{t \in T} q^i_t/2 \le 1-p^i_j$. The reason is that for every time $t \in T$ that vertex $i$ gets proposed to \ocs, it gets matched with the probability of at least $q^i_t/2$, so it will not be available at time $j$ with that probability. We also claim that $\sum_{t \in T} q^i_t/2 \le (1-q^i_j)/2$. The reason is that vertex $i$ can get matched at most once in every time interval of length $d$. Thus, vertex $i$ cannot get available twice in a time interval of length $d$, and we have $\sum_{t \in T} q^i_t + q^i_j \le 1$. This implies that $\sum_{t \in T} q^i_t/2 \le (1-q^i_j)/2$ and completes the proof of the claim.\hfill\Halmos
\end{proof}

Let $(i,j)$ be an edge in the graph. We consider the following cases.

\begin{itemize} [leftmargin=*]
    \item \textbf{Round $j$ is a randomized round, and $i \in Q_j$}: Since the algorithm did not choose to propose $i$ to \ocs\ in a deterministic round, the value of $\beta_j$ is at least $\Delta \beta^i_j$ when $i$ is proposed in a deterministic round. Therefore, we have
    \begin{align}
    \label{eq:an-case1-beta}
        \beta_j &\ge \Delta_{\texttt{deterministic}} \beta^i_j \ge \beta_1 \cdot r_i \cdot  q^i_j/2 + \beta_2 \cdot r_i \cdot (\Delta \mu^i_j - q^i_j/2) \nonumber& \\
        &\ge \beta_1 \cdot r_i \cdot q^i_j/2 + \beta_2  \cdot r_i (p^i_j - q^i_j/2) \,. \\
        \nonumber & \text{(Since in a deterministic round $\Delta \mu^i_j = p^i_j$).}
    \end{align}
    Also, since $i \in Q_j$, and the algorithm proposes $i$ in a randomized round, $\Delta \mu^i_j \ge p^i_j/2$. Therefore,
    \begin{align}
    \label{eq:an-case1-alpha}
        \alpha_{i,j} &= \Delta \alpha^i_j = (1-\beta_1) \cdot r_i \cdot q^i_j/2 + (1-\beta_2)\cdot r_i \cdot (\Delta \mu^i_j - q^i_j/2) \nonumber \\
        &\ge (1-\beta_1) \cdot r_i \cdot q^i_j/2 + (1-\beta_2) \cdot r_i \cdot (p^i_j/2 - q^i_j/2) \,.
    \end{align}
    Therefore, the LHS of the approximate duality constraint is equal to
    \begin{align*}
        &\beta_j + \sum_{t=\max\{j-d+1,1\}}^{j} \alpha_{i,t} = \overbrace{\beta_j}^{\text{By (\ref{eq:an-case1-beta})}} + \overbrace{\alpha_{i,j}}^{\text{By (\ref{eq:an-case1-alpha})}} + \overbrace{\sum\nolimits_{t=\max\{j-d+1,1\}}^{j-1} \alpha_{i,t}} ^{\text{By Observation \ref{obs:lb-dual}}}
        \\
        &\ge r_i \cdot \left(1-p^i_j/2+\beta_2(\varphi^i_j+3p^i_j/2-1) - \beta_1 \varphi^i_j\right) \,.
    \end{align*}
    Thus, the approximate duality constraint reduces to
    \begin{align}
    \label{eq:an-case1-const}
        1-p^i_j/2+\beta_2(\varphi^i_j+3p^i_j/2-1) - \beta_1 \varphi^i_j \ge \Gamma \,.
    \end{align}
    \item \textbf{Round $j$ is a deterministic round, and $i \in Q_j$}: Since $i \in Q_j$, and the algorithm proposes $i$ in a deterministic round, $\Delta \mu^i_j = p^i_j$. Therefore,
    \begin{align}
    \label{eq:an-case2-alpha}
        \alpha_{i,j} &= \Delta \alpha^i_j = (1-\beta_1) \cdot r_i \cdot q^i_j/2 + (1-\beta_2)\cdot r_i \cdot (\Delta \mu^i_j - q^i_j/2) \nonumber \\
        &\ge (1-\beta_1) \cdot r_i \cdot q^i_j/2 + (1-\beta_2) \cdot r_i \cdot (p^i_j - q^i_j/2) \,.
    \end{align}
     Therefore, the LHS of the approximate duality constraint is equal to
    \begin{align*}
        &\beta_j + \sum_{t=\max\{j-d+1,1\}}^{j} \alpha_{i,t} \ge \overbrace{\alpha_{i,j}}^{\text{By (\ref{eq:an-case2-alpha})}} + \overbrace{\sum\nolimits_{t=\max\{j-d+1,1\}}^{j-1} \alpha_{i,t}} ^{\text{By Observation \ref{obs:lb-dual}}}
        \\
        &\ge r_i \cdot \left(1 +\beta_2(\varphi^i_j+q^i_j/2-1) -\beta_1(\varphi^i_j+q^i_j/2)\right) \,.
    \end{align*}
        Thus, the approximate duality constraint reduces to
    \begin{align}
    \label{eq:an-case2-const}
         1 +\beta_2(\varphi^i_j+q^i_j/2-1) -\beta_1(\varphi^i_j+q^i_j/2) \ge \Gamma \,.
    \end{align}
    \item \textbf{Round $j$ is a randomized round, and $i \notin Q_j$}: Let $Q_j = \{i_1, i_2 \}$. Recall that the algorithm picks the set $Q_j$ that maximizes $\beta_j$. Therefore, $\Delta \beta^{i_1}_j$
    is at least $\Delta \beta^i_j$ for when we propose $Q'_j=\{i_2, i\}$ at the round $j$ instead of $Q_j$. Similarly, $\Delta \beta^{i_2}_j$
    is at least $\Delta \beta^i_j$ for when we propose $Q''_j=\{i_1, i\}$ at the round $j$ instead of $Q_j$. Let $j' = \prev^i_j$ be the previous occurrence of $i$. We further divide this case into two different cases.
    \begin{itemize} [leftmargin=12px]
        \item \textbf{When $i$ is not appeared in the last $d$ proposed sets to \ocs\, i.e., $j' \le j-d$:}
         In this case we have $q^i_j= p^i_j = 1$ , and $\Delta \beta^i_j$ is $r_i \cdot \beta_1/2$, when we propose $i$ in a randomized round. Thus, $\Delta \beta^{i_1}_j$ and $\Delta \beta^{i_2}_j$ are both at least $r_i \cdot \beta_1/2$.
        Therefore, the LHS of approximate duality is at least
        \begin{align}
        \label{eq:an-case3-beta-d}
        \beta_j + \sum_{t=\max\{j-d+1,1\}}^{j} \alpha_{i,t} \ge \beta_j \ge 2 (\beta_1/2) \cdot r_i = \beta_1 \cdot r_i  \,,
        \end{align}
        and, the approximate duality constraint reduces to
        \begin{align}
        \label{eq:an-case3-const1}
            \beta_1 \ge \Gamma \,.
        \end{align}
        \item \textbf{The other case is when $j' > j-d$:} In this case \ocs\ guarantees that $\Delta \mu^i_j$ is at least $(1/2+\gamma) u^i_j + q^i_j/2$ when we proposes $i$ in a randomized round. We can then give a lower bound on the value of $\beta_j$ as follows.
        \begin{align}
        \label{eq:an-case3-beta}
            \beta_j &\ge 2 r_i \cdot \big(\beta_1 q^i_j/2+ \beta_2 (1/2+\gamma) u^i_j\big)  = \beta_1 \cdot r_i \cdot q^i_j + \beta_2 \cdot r_i \cdot (1+2 \gamma) u^i_j \,.
        \end{align}
             Therefore, the LHS of the approximate duality constraint is equal to
    \begin{align*}
        &\beta_j + \sum_{t=\max\{j-d+1,1\}}^{j} \alpha_{i,t} \ge   \overbrace{\beta_j}^{\text{By (\ref{eq:an-case3-beta})}} + \overbrace{\sum\nolimits_{t=\max\{j-d+1,1\}}^{j-1} \alpha_{i,t}} ^{\text{By Observation \ref{obs:lb-dual}}}
        \\
        &\ge r_i \cdot \left(1-p^i_j+\beta_2\big((1+2\gamma)u^i_j+p^i_j+\varphi^i_j-1\big) + \beta_1(q^i_j- \varphi^i_j)\right) \,.
    \end{align*}
        Thus, the approximate duality constraint reduces to
    \begin{align}
        \label{eq:an-case3-const2}
       1-p^i_j+\beta_2\big((1+2\gamma)u^i_j+p^i_j+\varphi^i_j-1\big) + \beta_1(q^i_j- \varphi^i_j) \ge \Gamma \,.
    \end{align}
    \end{itemize}
    \item \textbf{Round $j$ is a deterministic round, and $i \notin Q_j$}:  Inspired by \cite{huang2020adwords} we introduce a set of new constraints to enforce the \textit{superiority} of randomized rounds. We then show that this set of new constraints can reduce this case to the previous case. Consider a vertex $i$ and round $j$, and let $\Delta_{\texttt{deterministic}} \beta^i_j$ be the $\Delta \beta^i_j$ when $i$ gets proposed in a deterministic round, and let $\Delta_{\texttt{randomized}} \beta^i_j$ be the $\Delta \beta^i_j$ when $i$ gets proposed in a randomized round. We then add the following set of constraints.
    \begin{align}
    \label{eq:superiority}
    2 \Delta_{\texttt{randomized}} \beta^i_j \ge \Delta_{\texttt{deterministic}} \beta^i_j \,.
    \end{align}
    Note that we can re-write the inequality above as following.
        \begin{align}
    &\beta_1 \cdot r_i \cdot q^i_j + \beta_2 \cdot r_i \cdot (1+2 \gamma) u^i_j \ge \beta_1 \cdot r_i \cdot q^i_j/2 + \beta_2 \cdot r_i \cdot  (p^i_j - q^i_j/2) \nonumber \\
    \label{eq:superiority-expanded}
    & \Rightarrow \beta_1  q^i_j + \beta_2 (1+2 \gamma) u^i_j \ge \beta_1  q^i_j/2 + \beta_2  (p^i_j - q^i_j/2) \,.
    \end{align}
    We now show that given these set of constraints, we can reduce this case to the previous case. Let $Q_j=\{i'\}$ be the deterministic proposal in round $j$. Since the algorithm did not propose the set $Q'_j=\{i,i'\}$ at this round, we can say that $\beta_j = \Delta_{\texttt{deterministic}} \beta^{i'}_j \ge  \Delta_{\texttt{randomized}} \beta^i_j + \Delta_{\texttt{randomized}} \beta^{i'}_j$. Also, by (\ref{eq:superiority}) we have $2 \Delta_{\texttt{randomized}} \beta^{i'}_j \ge \Delta_{\texttt{deterministic}} \beta^{i'}_j$. By combining this with the previous inequality we get $\Delta_{\texttt{randomized}} \beta^i_j \le \Delta_{\texttt{randomized}} \beta^{i'}_j$. Thus,
    \begin{align*}
    \beta_j \ge \Delta_{\texttt{randomized}} \beta^i_j + \Delta_{\texttt{randomized}} \beta^{i'}_j \ge 2\Delta_{\texttt{randomized}} \beta^i_j \,.
    \end{align*}
    This is the exact guarantee that we had for $\beta_j$ in the previous case. Thus, by adding the constraints (\ref{eq:superiority}) we can reduce this case to the previous case.
\end{itemize}

    \paragraph{\textbf{Optimizing the Competitive Ratio:}} We solve an LP whose variables are $\Gamma, \beta_1$, and $\beta_2$, in order to find the optimal competitive ratio.
    
    \begin{align*}
	\begin{array}{rl}
	\max & \Gamma \\
	\text{subject to} \\
	\forall p^i_j, q^i_j \in [0,1],  q^i_j \le p^i_j : & Eqn (\ref{eq:an-case1-const}), (\ref{eq:an-case2-const}),(\ref{eq:an-case3-const1}),(\ref{eq:an-case3-const2}),(\ref{eq:superiority-expanded}) \\
	& \beta_1 \ge \beta_2 \ge 0 \,.
	\end{array}
\end{align*}

 Although the constrainsts of the LP have variables like $\varphi^i_j$ and $u^i_j$, these variable are a function of $p^i_j$ and $q^i_j$. Thus, we only have to consider the constraints of LP for different values of $p^i_j$ and $q^i_j$. In fact it is straightforward to verify that $\Gamma= \beta_1 = \frac{3+ 4 \gamma}{6 + 6 \gamma}$ and $\beta_2 = \frac{1}{2+2 \gamma}$ is a feasible solution to the LP above. Note that we showed that an $1/32$-\ocs\ always exists.  This immediately implies that the competitive ratio of our algorithm is at least $\frac{3+ 4 \cdot 1/32}{6 + 6 \cdot 1/32} = 50/99 \approx 0.50505$.

%% file: tex/conclusion.tex
\section{Conclusion} 
We considered a fundamental generalization of classic bipartite online matching, where resources are reusable and used for an (identical) deterministic duration on every match. We motivated and introduced two new algorithms, Periodic Reranking, that strikes a careful balance between greedy and Ranking by reranking resources on a periodic schedule, and OCR-based primal-dual algorithm, that extends the powerful online correlated selection technique to reusable resources. We established a proof-of-concept result that we can obtain competitive ratios strictly better than $0.5$ with either of our algorithms. On our way to prove these results, we showed various novel structural properties of our algorithms, provided simpler understanding of existing results in the literature, and extended several tools to the case of non-reusable resources.


\paragraph{Open problems} Recall that a $(1-1/e)$ upper bound for OBMRR follows from the fact that OBMRR reduces to online bipartite matching (OBM) for large $d$. As OBMRR generalizes OBM, this upper bound continues to hold. The main open problem left unanswered in our paper is whether one can obtain the competitive ratio of $1-1/e$ for OBMRR.  It is worth noting that for small (but non-trivial) values of $d$, such as $d=2$, the upper bound of $(1-1/e)$ does not apply. Another important open problem is whether one can beat the competitive ratio of $0.5$ when different resources have different usage durations, or when usage durations are stochastic i.i.d. over time. One can also ask whether in these settings one can obtain an OCR algorithm with the types of performance guarantees we have in this paper. Finally consider two candidate algorithms: Ranking, which keeps the same randomized ranking for the resources at each time, and Reranking on Return (RoR), which reranks a resource every time it returns back to the system after a match. Notice that \alg\ generates a new rank more frequently than RoR. In RoR, if a resource is not highly ranked then it may not be matched and its rank is not reset. Consequently, RoR does not fully succeed in untangling dependence between matching decisions at arrivals that are well separated.  Analyzing the performance of either RoR or Ranking remain as challenging open problems.

	

%% file: tex/ack.tex
\section*{Acknowledgments.}
The authors thank the anonymous EC referees for their careful comments and feedback. The authors also thank Vineet Goyal and Garud Iyengar for many insightful discussions on this topic.

%% file: tex/appendix.tex
\begin{APPENDICES}
\renewcommand{\theHsection}{A\arabic{section}}

\input{tex/apx-cont}

\input{tex/apx-reverse}

\input{tex/apx-missing}

 \end{APPENDICES}

%% file: tex/apx-cont.tex
\section{Continuous Time Model}\label{appx:cont}

 In a continuous time version of the problem, we have an increasing function $a:U\to \mathbb{R}_{+}$ that maps arrivals $j\in U$ to arrival times $a(j)$. The difference $a(j+1)-a(j)$ can be arbitrarily small. If a resource is matched to arrival $j$, it is unavailable during the time interval $(a(j),a(j)+d]$. There can be an arbitrary number of arrivals in this interval. In contrast, in the discrete time model a resource in use is unavailable for exactly $d-1$ arrivals. For simplicity, we focus on the discrete time model but our algorithms and performance guarantees translate as is to the continuous time model. In fact, by adding dummy arrivals and adjusting time scales, we can cast an instance of the continuous model with $|U|$ arrivals and usage duration $d$ into an instance of the discrete model with $|U|^2$ arrivals and usage duration $|U|$. 

Recall that a $(1-1/e)$ upper bound for OBMRR follows from the fact that OBMRR reduces to OBM for large $d$. As OBMRR generalizes OBM, this upper bound continues to hold. However, for small (but non-trivial) values of $d$, such as $d=2$, the upper bound of OBM does not apply. For the continuous time model, this upper bound holds for all non-zero values of $d$. To see this, we use the fact that there can be an arbitrary number of arrivals in any given time duration to fit a hard instance(s) of OBM into a duration smaller than $d$. More generally, this implies that the continuous model is equally hard for all finite and non-zero $d$ i.e., if there is an upper bound of $\alpha$ for duration $d'$, one can scale the hard instances to obtain an $\alpha$ upper bound for any other (finite and non-zero) duration $d$. 

%% file: tex/apx-reverse.tex
\section{Proof of \Cref{clm:reverse}}
\label{apx:claim-rev}
In order to prove the Claim \ref{clm:reverse}, we first have to define the \textit{dependency graph}.
\paragraph{Dependency Graph.} Let $Q_1, Q_2, \cdots, Q_{\ell}$ be the sequence of the sets  proposed by the outer algorithm, we define the dependency graph $G=(V,E)$ as follow. In the dependency graph $G$ we have a vertex for every set $Q_j$, i.e., $V=\big \{Q_j : j \in [\ell]\big \}$. Also, for every offline vertex $i$ we have an edge $(Q_j,Q_{j'})_i$ between two sets $Q_{j}$ and $Q_{j'}$ if both $Q_j$ and $Q_{j'}$ are randomized rounds, $i \in Q_{j}, Q_{j'}$, and $i$ is not in any set proposed to the algorithm between the times $j$ and $j'$. We use subscript $i$ in the edge $(Q_j,Q_{j'})_i$ to distinguish parallel edges. 

Consider an edge $(Q_j,Q_{j'})_i$ in the dependency graph, and by symmetry assume that  $Q_j$ is proposed to \ocs\ algorithm before $Q_{j'}$. In this case the decision of \ocs\ for rounds $Q_j$ and $Q_{j'}$ is perfectly negatively correlated, and it selects $i$ in exactly one of $Q_j$ or $Q_{j'}$ if (i) round $Q_j$ is a \textit{sender} round, (ii) \ocs\ selects $i_m = i$ in round $Q_j$, (iii) round $Q_{j'}$ is a \textit{receiver} round, (iv) \ocs\ selects $i_m=i$ in round $Q_{j'}$. However, this perfect negative correlation is not always guaranteed, and conditions (i)-(iv) should all be met to have a negative correlation. In this case, we say edge $(Q_j, Q_{j'})_i$ is \textit{realized} and is in the \textit{realized dependency graph}.

For each edge $(Q_j,Q_{j'})_i$ in the dependency graph, there is a probability of $1/16$ that this edge gets realized in the realized dependency graph. The reason is that \ocs\ satisfies conditions (i)-(iv) each with the probability of $1/2$. Moreover, it is easy to verify that every vertex in the realized dependency graph have at most one realized edge. Consider a \textit{sender} round $Q_j$, and let $i_m=i$ be the vertex selected by \ocs\ in order to provide a negative correlation. In this case the corresponding vertex of $Q_j$ in the realized dependency graph can have at most one edge in the form of $(Q_j, Q_{j'})_i$ where $Q_{j'}$ is the next time that vertex $i$ is proposed to the algorithm. We can use a similar argument to show that every \textit{receiver} round has at most one realized edge as well. Therefore, every vertex has a degree of at most one in the realized dependency graph, i.e., realized edges form a matching between vertices of the dependency graph. 

We are now ready to prove Claim \ref{clm:reverse}.

\begin{proof} [Proof of Claim \ref{clm:reverse}]
We call the input sequence $Q_1, Q_2, \cdots, Q_{\ell}$ the \textit{forward} sequence, and the input sequence $Q_{\ell}, Q_{\ell-1}, \cdots, Q_1$ the \textit{backward} sequence. For any offline vertex $i$ we show that expected number of times the algorithm matches $i$ in the forward and backward sequence is the same.  Let $\vec{X}=\langle X_1, X_2, \cdots, X_{\ell} \rangle$ be random variables where $X_j$ is $1$ when \ocs\ selects $i$ when it receives $Q_j$ in the forward sequence, and is $0$ otherwise. Similarly, let $\vec{Y} =\langle Y_1, Y_2, \cdots, Y_{\ell} \rangle$ be random variables where $Y_j$ indicates whether \ocs\ has selected $i$ when it receives $Q_j$ (at the time $\ell-j+1$) in the backward sequence. In order to prove the claim we show that the joint distribution of $\vec{X}$ random variables is the same as $\vec{Y}$ random variables. In other words, let $\vec{b} \in \{0,1\}^{\ell}$ be any binary vector of length $|U|$, we show the following.
\begin{align}
\label{eq:joint}
    \Pr\big[\, \vec{X}=\vec{b}\, \big] = \Pr\big[\, \vec{Y}=\vec{b}\, \big] \,.
\end{align}
We first show how (\ref{eq:joint}) directly implies Claim \ref{clm:reverse}. Let $\mu_{\rhd}$ and $\mu_{\lhd}$ be random variables which indicate the number of times the algorithm matches $i$ in the forward and backward sequences respectively. Consider a vector $\vec{b} \in \{0,1\}^{\ell}$, assuming (\ref{eq:joint}), $\vec{X}$ and $\vec{Y}$ are equal to $\vec{b}$ with the same probability. We first show that if the $\vec{X}$ and $\vec{Y}$ are both equal to a vector $\vec{b}$, then the algorithm matches vertex $i$ exactly the same number of times in forward and backward sequences. In specific we show the following.

\begin{observation}
\label{obs:greedysym}
For any vector $\vec{b} \in \{0,1\}^{\ell}$, we have
\begin{align*}
    \E \big[\, \mu_{\rhd} \ | \vec{X} = \vec{b}\, \big ]  = \E \big[\, \mu_{\lhd} \, | \vec{Y} = \vec{b}\, \big] \,.
\end{align*}
\end{observation}

\begin{proof}
    Recall that for an offline vertex $i$, the way that the outer algorithm constructs the matching using the decisions of \ocs\ is that it adds $i$ to the matching whenever it is available and get selected by \ocs. Let $\vec{b} = \langle b_1, b_2, \cdots, b_{\ell} \rangle$, and $0 \le j_1 <j_2 < \cdots < j_k$ be all indices such that $b_{j_{k'}}$ is 1 where $k' \in [k]$. By conditioning on $\vec{X} = \vec{b}$, we can say that the \ocs\ algorithm selects $i$ when it receives $Q_{j_1}, Q_{j_2}, \cdots, Q_{j_k}$. In the forward sequence, the algorithm first matches $i$ when it proposes $Q_{j_1}$, and then vertex $i$ gets available after $d$ time units. Therefore, the next time the algorithm matches $i$ is the smallest index $j' \in \{j_1, \cdots, j_k\}$ such that $j' \ge j_1 + d$, and the algorithm continues matching $i$ in the same way. It is easy to verify that the number of times that our algorithm matches $i$ is equal to the maximum number of times that we can match $i$ at the time units  $\{j_1, \cdots, j_k\}$ without violating the reusability constraints (i.e., we should not match $i$ twice within any time interval of length $d$).
    
    Now consider the backward sequence. By conditioning on $\vec{Y}=\vec{b}$, we can say that \ocs\ selects $i$ in exactly the same proposed sets $Q_{j_1}, Q_{j_2}, \cdots, Q_{j_k}$. However, in the backward sequence the first time that algorithm matches $i$ is when it proposes $Q_{j_k}$, and then it cannot match $i$ again when it proposes sets $Q_{[j_k,j_k-d+1]}$. Similar to the forward sequence, it is easy to see that the number of times that the algorithm matches $i$ is equal to the maximum number of times that we can match $i$ without violating the reusability constraints. Therefore, the expected size of the matching in the forward and backward sequences is the same.
\end{proof}

Using the observation above and assuming (\ref{eq:joint}), we can easily prove Claim \ref{clm:reverse} as follows.
\begin{align*}
    \mu^i(Q_{[1,\ell]}) & = \E[\mu_\rhd] \\
    & = \sum_{\vec{b} \in \{0,1\}^{\ell}}  \E \big[\, \mu_{\rhd} \, | \vec{X} = \vec{b}\, \big] \Pr\big[\, \vec{X}=\vec{b}\, \big]
    \\& = \sum_{\vec{b} \in \{0,1\}^{\ell}}  \E \big[\, \mu_{\rhd} \, | \vec{X} = \vec{b}\, \big] \Pr\big[\, \vec{Y}=\vec{b}\, \big] & \text{By (\ref{eq:joint}).}
    \\& = \sum_{\vec{b} \in \{0,1\}^{\ell}}  \E \big[\, \mu_{\lhd} \, | \vec{Y} = \vec{b}\, \big] \Pr\big[\, \vec{Y}=\vec{b}\, \big] & \text{By Observation \ref{obs:greedysym}.}
    \\& = \E[\mu_\lhd] = \mu^i(Q_{[\ell,1]}) \,, 
\end{align*}
which implies the claim. Therefore, in order to show the correctness of Claim \ref{clm:reverse}, we only need to prove Equation (\ref{eq:joint}).

Let $G=(V,E)$ be the dependency graph for the forward sequence. It is easy to see that the dependency graph for the backward sequence is $G$ as well. Recall that in the dependency graph, we have a vertex for every proposed set. Since both forward and backward sequences have the same proposed sets, they have the same vertices in the dependency graph. Furthermore, they have the same set of edges in the dependency graph. As we discussed earlier we add an edge $(Q_j,Q_{j'})_i$ (in the forward sequence) between two sets $Q_{j}$ and $Q_{j'}$ if both $Q_j$ and $Q_{j'}$ are randomized rounds, $i \in Q_{j}, Q_{j'}$, and $i$ is not in any set proposed to the algorithm between $Q_j$ and $Q_{j'}$. In that case $(Q_j,Q_{j'})_i$ will be in the dependency graph of the backward sequence as well, since $i$ is proposed in $Q_j$ and $Q_{j'}$ but not in any set in-between them.

We further claim that the probability distribution of \textit{realized} dependency graph is also the same for forward and backward sequences. Let $E_\rhd$, $E_\lhd$ be the set of realized edges in the forward and backward sequences, respectively. We specifically show that for any set of edges $S \subseteq E$, the realized edges $E_\rhd$ and $E_\lhd$ will be equal to $S$ with the same probability. I.e.,
\begin{align}
\label{eq:probability-realized}
    \Pr\big[E_\rhd = S\big] = \Pr\big[E_\lhd = S\big] \,.
\end{align}

Consider an edge $(Q_j,Q_{j'})_i$ in the dependency graph where $j<j'$. In this case this edge is in the realized dependency graph of the forward sequence if (i) round $Q_j$ is a \textit{sender} round, (ii) \ocs\ selects $i_m = i$ in round $Q_j$, (iii) round $Q_{j'}$ is a \textit{receiver} round, (iv) \ocs\ selects $i_m=i$ in round $Q_{j'}$. As we discussed earlier \ocs\ satisfies each of conditions (i)-(iv) each with the probability of 1/2. Note that conditions (i) and (ii) are related to decisions of the algorithm when it receives $Q_j$, and with the probability of $1/4$ the algorithm satisfies these conditions. Also  conditions (iii) and (iv) are only related to decisions of the algorithm when it receives $Q_{j'}$ which they met with the probability of $1/4$. Another way to interpret the realization of dependency graph is as follow. Vertex $Q_j$ has at most $4$ edges in the dependency graph, and it satisfies the conditions (i)-(iv) for at most one of them each with the probability of $1/4$. The same thing holds for vertex $Q_{j'}$. It satisfies the conditions for at most one of its incident edges each with the probability of $1/4$. Therefore, we get the same distribution of realized dependency graphs if each vertex selects at most one of its incident edges each with the probability of $1/4$, and we add an edge to the realized graph if it has been selected by both of its endpoints.

Now consider the same edge $(Q_j,Q_{j'})_i$ in the dependency graph of the backward sequence. In this case the algorithm receives the proposed set $Q_{j'}$ before $Q_j$. Therefore, this edge is in the realized dependency graph of the backward sequence if (i) round $Q_{j'}$ is a \textit{sender} round, (ii) \ocs\ selects $i_m = i$ in round $Q_{j'}$, (iii) round $Q_{j}$ is a \textit{receiver} round, (iv) \ocs\ selects $i_m=i$ in round $Q_{j}$. We can still interpret the realization of the dependency graph as if each vertex selects at most one of its incident edges each with the probability of $1/4$, and we add an edge to the realized graph if it has been selected by both of its endpoints. Therefore, the probability of distribution of dependency graphs are the same for both forward and backward sequences which implies (\ref{eq:probability-realized}).

As we discussed the distribution the realized dependency graph is the same for both forward and backward sequences. We further claim that for every offline vertex $i$ the proposed sets that \ocs\ selects $i$ also have the same distribution for the forward and backward sequences. As we discussed earlier, edges in the realized dependency graph form a matching, and each edge represents a perfect negative correlation between the decisions of \ocs. 
We first claim that in order to find the joint distribution of rounds in which the algorithm picks $i$, we can only look at the decisions of the algorithm for rounds that $i$ is proposed. If we consider a 
randomized round $Q_j=\{i, i'\}$, the algorithm picks each of $i$ and $i'$ with the probability of $1/2$ when the probability is taken over all random decisions of the algorithm. Algorithm \ref{alg:ocs-sym} just introduces slightly negative correlation between its decision for different rounds. Now consider an edge $(Q_j,Q_{j'})_{i'}$ in the realized dependency graph where $Q_{j'}$ does not contain $i$. Although the decision of the algorithm has a perfect negative correlation the rounds $Q_{j'}$ and $Q_j$, vertex $i$ still gets selected by the algorithm with the probability of $1/2$ at round $Q_j$. Also, the decision of the algorithm for the round $Q_j$ is independent of its decisions for other rounds that $i$ is proposed.

It follows that the distribution of the algorithm for selecting vertex $i$ is as follows. 
Consider a randomized round $Q_j \ni i$. If $Q_j$ has no incident edges in the realized dependency graph, then the algorithm selects $i$ with the probability of $1/2$ independently. Otherwise, let assume that there is an edge between $Q_j$ and $Q_{j'}$ in the dependency graph. If $i \notin Q_{j'}$, the algorithm still selects $i$ he probability of $1/2$ and this decision is independent of other rounds that $i$ is proposed. If $i \in Q_{j'}$, the algorithm selects $i$ in exactly one of $Q_j$ and $Q_{j'}$ each with the probability of $1/2$.
Furthermore, for a deterministic round $Q_j \ni i$, the algorithm selects $i$ with the probability of $1$. It follows that the distribution of the rounds that the algorithm selects $i$ is the same for both forward and backward sequences if they share a same dependency graph. In other words, for every set of edges $S \subseteq E$ in the realized graph and any vector $\vec{b} \in \{0,1\}^{\ell}$, we have
\begin{align}
    \label{eq:samedisQj}
    \Pr\big[\vec{X}= \vec{b} \, | E_\rhd=S\big] = \Pr\big[\vec{Y}= \vec{b} \, | E_\lhd=S\big] \,.
\end{align}
Therefore we have
\begin{align*}
    \Pr\big[\vec{X}= \vec{b}\, \big] &= \sum_{S \subseteq E}  \Pr\big[\vec{X}= \vec{b} \, | E_\rhd=S\big] \Pr\big[E_\rhd=S\big] \\
    &= \sum_{S \subseteq E} \Pr\big[\vec{X}= \vec{b} \, | E_\rhd=S\big] \Pr\big[E_\lhd=S\big] & \text{By (\ref{eq:probability-realized}).} \\
    &= \sum_{S \subseteq E} \Pr\big[\vec{Y}= \vec{b} \, | E_\lhd=S\big] \Pr\big[E_\lhd=S\big] & \text{By (\ref{eq:samedisQj}).} \\
    &=  \Pr\big[\vec{Y}= \vec{b}\, \big] \,,
\end{align*}
which proves (\ref{eq:joint}) and completes the proof of the Claim \ref{clm:reverse}.
\end{proof}

%% file: tex/apx-missing.tex
\section{Missing Proofs}\label{appx:missing}
\begin{proof}{Proof of Lemma \ref{lem:primal-dual}.}
$\{\alpha_{i,t}/\Gamma\}_{i,t}$, $\{\beta_j/\Gamma\}_{j}$ is a feasible dual assignment, due to the approximate dual feasibility. Denoting the optimal LP objective value by $\textsc{OPT}$, by weak duality we have $D/\Gamma\geq \textsc{OPT}$. Combined with the reverse weak duality, we have $P\geq \Gamma\cdot\textsc{OPT}$, as desired.\hfill\Halmos
\end{proof}
\medskip

\begin{proof}{Proof of \Cref{devlb}.} 
		
		\emph{Part(a)}: By definition of $y^c_j(y^1)$, we have $\beta_j(y^1,1)\, \geq\, r_i \left(1-g(y^c_j(y^1))\right)$. It remains to show that, $\beta_j(y^1, y^2)\,\geq\,\beta_j(y^1,1)$. Let $r_t(y^1,y^2)$ denote the reduced price of the resource matched to arrival $t\in U$ in the matching $\alg(y^1,y^2)$. Set $r_t(y^1,y^2)=0$ if $t$ is unmatched.

		Since $g(x)$ is strictly increasing in $x$ (for $\beta>0$), 
		it suffices to show that 
		\[r_t(y^1,y^2)\geq r_t(y^1,1)\,\,\, \forall\, y^2\in[0,1],\,\, t\in k(j).\] 	
		Note that $1-g(1)=0$ for every $\beta$. Therefore, when $y^2=1$ the reduced price of arrival matched to $i$ is 0, same as, if the arrival were unmatched. Combining this observation with the fact that \alg\ matches each arrival greedily based on reduced prices, it suffices to show that,
		\begin{equation}\label{neste}
			S_t(y^1,1)\backslash\{i\}\subseteq S_t(y^1,y^2)\quad \forall y^2\in[0,1],\, t\in k(j).
		\end{equation}

		We prove \eqref{neste} via induction over arrivals in period $k(j)$. Let $t(y^2)$ denote the first arrival in period $k(j)$ where $i$ is available. Matchings $\alg(y^1,1)$ and $\alg(y^1,y^2)$ are identical prior to $t(y^2)$. Thus, $S_{t(y^2)}(y^1,1)= S_{t(y^2)}(y^1,y^2)$. Now, suppose that \eqref{neste} holds for all arrivals $t< t'\in k(j)$. We show that \eqref{neste} holds for arrival $t'$ as well.
		
		For the sake of contradiction, suppose there exists a resource $v\in S_{t'}(y^1,1)\backslash (S_{t'}(y^1,y^2)\cup \{i\})$.	Recall that $S_{t}(y^1,1)\backslash\{i\}\subseteq S_t(y^1,y^2)$ for all $t<t'$. 
		This occurs only if $v$ is matched to arrival $t'-1$ in $\alg(y^1,y^2)$, where $t'-1\geq t(y^2)$. 
		Since every resource is matched at most once in a period, we have, $S_{t'}(y^1,1)\subseteq S_{t'-1}(y^1,1)$. Thus, $v\in S_{t'-1}(y^1,1)\backslash\{i\}\subseteq S_{t'-1}(y^1,y^2)$ i.e., in $\alg(y^1,1)$, resource $v$ is available but not matched to $t'-1$. This contradicts the fact that \alg\ matches greedily based on reduced prices.
		\smallskip
		
		\noindent \emph{Part(b)}: Let $T'$ denote the set of arrivals in the interval $k(j)\cap \{j-d+1,\cdots,j\}$. We are given that resource is available at some point in this interval. Further, every resource is matched to at most one arrival in $T'$. Consider an arbitrary value $y^2=z\in[0,1]$ and the matching $\alg(y^1,z)$.

		First, if $i$ is matched to an arrival in $T'$, then by definition of dual variables \eqref{theta}, we have that $\sum_{t=\max\{j-d+1,1\}}^{j} \alpha_{it}(y^1,z) \geq r_i g(z)$. On the other hand, if $i$ is not matched to any arrival in $T'$, then, $\alg(y^1,z)$ is identical to $\alg(y^1,1)$ until after arrival $j$. Therefore, arrival $j$ is matched to the same resource in both in $\alg(y^1,z)$ and $\alg(y^1,1)$, despite $i$ being available. Thus, 
		$r_i(1-g(z))<r_i(1-g(y^c_j(y^1)))$ i.e., $z>y^c_j(y^1)$.
		\hfill\Halmos\end{proof}
		\medskip
		
		\begin{proof}{Proof of Claim \ref{claim:receiversender}.}
    We prove the claim using an induction on the size of $\vec{Q}$. When $|\vec{Q}|=1$, there is no difference between sender and receiver rounds and the claim clearly holds. Now suppose that $|\vec{Q}|>1$. Let $j' = \prev^i_j$, then by (\ref{eq:musender}) we have
    \begin{align*}
                 \mu^i_{\sender}(\vec{Q}_{[j,1]})& \ge \frac{1}{2} \cdot \big(1+ \mu^i(\vec{Q}_{[j-d,1]})\big) + \frac{1}{4} \cdot \mu^i(\vec{Q}_{[j-1,1]}) \nonumber \\
             &+  \frac{1}{8} \cdot \mu^i_{\sender}(\vec{Q}_{[j',1]}) + \frac{1}{16} \cdot \mu^i_{\receiver}(\vec{Q}_{[j',1]}) + \frac{1}{16} \cdot \big(1+ \mu^i(\vec{Q}_{[j'-d,1]})\big)
    \end{align*}
            Note that by Equation (\ref{eq:mu-rs}), we have $\mu^i(\vec{Q}_{[j-1,1]}) = \mu^i(\vec{Q}_{[j',1]}) = \frac{1}{2} \cdot \mu^i_{\sender}(\vec{Q}_{[j',1]}) + \frac{1}{2} \cdot \mu^i_{\receiver} (\vec{Q}_{[j',1]})$. Thus,
            \begin{align*}
            \mu^i_{\sender}(\vec{Q}_{[j,1]})& \ge \frac{1}{2} \cdot (1+ \mu^i(\vec{Q}_{[j-d,1]}))+ \frac{1}{2} \cdot \mu^i(\vec{Q}_{[j-1,1]}) \\
            & + \frac{1}{16} \cdot \big(1+ \mu^i(\vec{Q}_{[j'-d,1]}) - \mu^i_{\receiver}(\vec{Q}_{[j',1]})\big) \\
            & \ge \mu^i_{\receiver}(\vec{Q}_{[j,1]}) + \frac{1}{16} \cdot \big(1+ \mu^i(\vec{Q}_{[j'-d,1]}) - \mu^i_{\receiver}(\vec{Q}_{[j',1]})\big) & \text{By (\ref{eq:mureciver}).} \\
            \end{align*}
            Note that we always have $1+ \mu^i(\vec{Q}_{[j'-d,1]}) \ge \mu^i(\vec{Q}_{[j',1]})$, since $1+ \mu^i(\vec{Q}_{[j'-d,1]})$ describes an event that $i$ got matched at $Q_{j'}$ deterministically. Thus,
             \begin{align*}
            \mu^i_{\sender}(\vec{Q}_{[j,1]})& \ge \mu^i_{\receiver}(\vec{Q}_{[j,1]}) + \frac{1}{16} \cdot \big(\mu^i(\vec{Q}_{[j',1]}) - \mu^i_{\receiver}(\vec{Q}_{[j',1]})\big) \,. \\
            \end{align*}
            Also, by the induction hypothesis we have $\mu_i(\vec{Q}_{[j',1]}) - \mu^i_{\receiver}(\vec{Q}_{[j',1]}) \ge 0$. Thus, using the inequality above we can say that $\mu^i_{\sender}(\vec{Q}_{[j,1]}) \ge  \mu^i_{\receiver}(\vec{Q}_{[j,1]})$. By combining this with (\ref{eq:mu-rs}) we get $\mu^i_{\sender}(\vec{Q}_{[j,1]}) \ge \mu^i(\vec{Q}_{[j,1]}) \ge  \mu^i_{\receiver}(\vec{Q}_{[j,1]})$ which completes the proof for the induction, and proves the claim.\hfill\Halmos
    \end{proof}
    \medskip